\documentclass{lmcs} 
\pdfoutput=1
\usepackage[utf8]{inputenc}

\usepackage{lastpage}
\lmcsdoi{20}{2}{12}
\lmcsheading{}{\pageref{LastPage}}{}{}%
{Nov.~30,~2022}{Jun.~12,~2024}{}


\usepackage{hyperref}
\usepackage{macro}
\theoremstyle{plain} 

\usepackage{amsmath,amsfonts, tikz, amssymb, amsthm}
\usetikzlibrary{petri, positioning,automata}
\usetikzlibrary{arrows}
\usepackage{algorithm}
\usepackage{algorithmic}
\makeatletter
\newenvironment{procedure}[1][htb]{%
	\renewcommand{\ALG@name}{Procedure}
	\begin{algorithm}[#1]%
	}{\end{algorithm}}
\makeatother
\usepackage{cancel}
\usepackage{xspace}
\usepackage{multicol}

\begin{document}
		
		\title[Branch-WSTS]{Branch-Well-Structured Transition\texorpdfstring{\\}{} Systems and Extensions}
		
		\author[B.~Bollig]{Benedikt Bollig\lmcsorcid{0000-0003-0985-6115}}[a]
		\author[A.~Finkel]{Alain Finkel\lmcsorcid{0000-0003-0702-3232}}[a,b]
		\author[A.~Suresh]{Amrita Suresh\lmcsorcid{0000-0001-6819-9093}}[a,c]
		
		\address{Universit\'e Paris-Saclay, ENS Paris-Saclay, CNRS, LMF, 91190, Gif-sur-Yvette, France}	
		\email{\{benedikt.bollig,~alain.finkel\}@ens-paris-saclay.fr}  
		
		\address{Institut Universitaire de France}	
		\address{University of Oxford, United Kingdom}
		\email{amrita.suresh@cs.ox.ac.uk}
		
		
		
		
		
		\begin{abstract}
			\noindent We propose a relaxation to the definition of well-structured transition systems (\WSTS) while retaining the decidability of boundedness and non-termination. In this class, the well-quasi-ordered (\wqo) condition is relaxed such that it is applicable only between states that are reachable one from another. Furthermore, the monotony condition is relaxed in the same way. While this retains the decidability of non-termination and boundedness, it appears that the coverability problem is undecidable. To this end, we define a new notion of monotony, called cover-monotony, which is strictly more general than the usual monotony and still allows us to decide a restricted form of the coverability problem. 
		\end{abstract}
		
	\maketitle

		
		\section{Introduction}
		A well-structured transition system (\WSTS), initially called structured transition system in \cite{finkel_generalization_1987, finkel_reduction_1990}, is a mathematical framework used to model and analyse the behaviour of concurrent and reactive systems. It provides a formal representation of the system's states, and the transitions between them. The key characteristic of a \WSTS is that it exhibits certain properties that enable algorithmic verification of important system characteristics, such as non-termination, boundedness, and coverability.

		In a \WSTS, the system is represented as a (potentially infinite) set of \emph{states}, denoted by $\stateset$ equipped with a \emph{transition relation} ${\tranrel} \subseteq \stateset \times \stateset$. Here, $\stateset$ may represent different configurations, states of a computation, or states of a concurrent program, depending on the specific context of the system being modelled. Furthermore, the states in $\stateset$ are partially ordered by a quasi-ordering relation, denoted by $\leq$, and the transition relation $\tranrel$ fulfils one of various possible monotonies with respect to $\leq$. To be considered a well-structured transition system, the partial order $\leq$ on the set of states $\stateset$ must be \emph{well}, i.e. well-founded with no infinite antichains (Section~\ref{sec:prelim} formally defines these notions).
		
		The theory of WSTS, although initially formulated in the late '80s \cite{finkel_generalization_1987}, is still being widely explored (refer to \cite{finkel_well-structured_2001,abdulla_algorithmic_2000, ABDULLA2011248} for surveys). Some notable recent developments in the theory include handling infinite branching \cite{blondin_handlin_2018}, ideal abstraction \cite{zufferey_ideal_2012}, extensions over tree behaviours \cite{schmitz_branching_2021,lazic_nonelementary_2015}, etc. Because they provide a powerful framework for reasoning about infinite-state systems and enable the development of effective algorithms for automatic verification, WSTS have found applications in various areas, including classically, Petri Nets and Lossy Channel Systems, but also recent explorations in graph transformation systems \cite{ozkan_decidability_2022}, program verification 
		\cite{finkel_verification_2020}, $\mu$-calculus \cite{shilov_well-structured_2007, kouzmin_model_2004, bertrand_computable_2013}, cryptographic protocols \cite{dosualdo_decidable_2020}, Presburger counter machines \cite{finkel_well_2019}, among others.
		
		Since its development, a number of relaxations and extensions of the two core assumptions of WSTS have been studied. Notably, with regards to the monotony, the notion of transitive, stuttering monotonies has been studied \cite{finkel_well-structured_2001}. With regards to the well-quasi-order assumption, \cite{blondin_well_2017} has shown that coverability is still decidable without the assumption of wellness, with only the absence of infinite antichains. They refer to this class as \emph{well-behaved transition systems} (\WBTS). However, in the same paper, they also show that this relaxation is too broad for the decidability of boundedness and non-termination, and show classes of \WBTS with undecidable boundedness and non-termination properties. 
		
		\smallskip
		In this work, we look at relaxing both the well-quasi-order (\wqo) and the monotony assumptions, while still ensuring decidable non-termination and boundedness. More precisely, we introduce the notion of \emph{branch-well-quasi-order} (branch-\wqo), which is only applicable to states that are reachable one from another. Interestingly, branch-\wqo share many of the closure properties that \wqo enjoy. Similarly, we also relax the notion of monotony to \emph{branch-monotony}, which requires only that the same sequence of actions can be repeated from two states ordered by $\leq$. We call the systems which satisfy these two conditions \emph{branch-\WSTS}. Branch-WSTS, though a strict superclass of WSTS, still has decidable boundedness and non-termination, under the same effectivity conditions.
		
		We show examples of classes of systems studied in the literature, which are not WSTS, but are branch-WSTS. Notably, there are non-trivial subclasses in both counter machines and FIFO automata (two classes that are well-known to be Turing-hard) which are not WSTS but are branch-WSTS. Moreover, we show that coverability is, in fact, undecidable for general branch-WSTS. This shows that WBTS and branch-WSTS are incomparable. Furthermore, we show that the properties of branch-monotony can be verified along a single run in order to show non-termination and unboundedness. This provides an algorithm to verify these properties for a large subclass of counter machines and branch-\wqo FIFO machines. We also see that for FIFO machines and the prefix-ordering, the conditions that relate to the validity of branch-monotony along a single run imply infinite iterability, as defined in \cite{finkel_verification_2020}.
		
		The other major contribution of this work deals with coverability. Contrary to \cite{blondin_well_2017}, we relax the notion of monotony, which we call \emph{cover-monotony}, while maintaining the well-quasi-order assumption, such that the resulting class of systems have decidable coverability. For this class, which we refer to as \emph{cover-WSTS}, we show that a restricted version of the coverability problem is decidable, even in the absence of strong (or strict or transitive or reflexive) monotony. 
	
\smallskip	
		A preliminary version of this work has been presented at the 42nd International Conference on Formal Techniques for Distributed Objects, Components, and Systems (FORTE 2022). The contributions in this work extend the conference version as follows:
		This work further investigates the classes of branch-WSTS. We explicitly prove that counter machines with restricted zero tests are branch-monotone. Furthermore, in this work, we prove the conjecture that was stated in the conference version claiming normalized input-bounded FIFO machines are branch-WSTS. We also provide a sufficient condition for boundedness and non-termination for a large subclass of counter machines and FIFO machines, which strictly include branch-\WSTS. Apart from this, we also provide explicit proofs and examples of systems which were not included in the conference version, along with some additional properties, notably for branch-\wqos. 
		
		\smallskip
		\noindent\textbf{Outline.}~ Section~\ref{sec:prelim} introduces terminology and well-known results concerning well-quasi-orderings and well-structured transition systems. Section~\ref{sec:termbound} defines branch-WSTS, and shows that both the boundedness and the non-termination problems are decidable for such systems.
		Section~\ref{sec:examples} provides some examples of branch-\WSTS as well as provides a sufficient condition for non-termination for a large class of counter and FIFO machines (strictly including branch-\WSTS).
		Section~\ref{sec: coverability} investigates the coverability problem for WSTS with relaxed conditions. We conclude in Section~\ref{sec:concl}.	
		
		\section{Preliminaries} \label{sec:prelim}
		
		We denote the set of natural numbers by $\nat$. We denote the standard ordering on $\nat$ by $\leq$. Given a set $A$, we denote by $|A|$ the number of distinct elements in $A$.
		
		Let $\alphset$ be a finite alphabet, and $\alphset^*$ be the set of finite words over $\alphset$.
		We denote the empty word by $\emptyword$. The concatenation of two words $u, v \in \alphset$ is denoted by $u \cdot v$, or $u.v$, or simply $uv$. Given $a \in \alphset$ and $w \in \alphset^*$, we let $|w|_a$ denote the number of occurrences of $a$ in $w$. With this, we let $\Alphs{w} = \{a \in \alphset \mid |w|_a \geq 1\}$. A word $u \in \alphset^*$ is a prefix (resp. suffix) of $w \in \alphset^*$ if $w=u\cdot v$ (resp. $w=v\cdot u$) for some $v\in \alphset^*$. The sets of prefixes and suffixes of $w \in \alphset^*$ are denoted by $\Pref{w}$ and $\Suf{w}$, respectively. We denote the prefix ordering on words over an alphabet $\alphset$ by $\pref$. More specifically, for two words $u, w \in \alphset^*$, we say $u \pref w$ if $u$ is a prefix of $w$.
		
		\subsection{Well-structured transition systems}
		We define some preliminary notions associated with well-structured transition systems.
		
		\smallskip
		\noindent\textbf{Quasi-orderings.}~
		Let $\stateset$ be a set and ${\leq} \subseteq \stateset \times \stateset$ be a binary relation over $\stateset$, which we also write as $(\stateset,\leq)$. We call $\leq$ a \emph{quasi ordering} (qo) if it is reflexive and transitive. As usual, we call $\leq$ a \emph{partial ordering} if it is a quasi-ordering and anti-symmetric (if $x \leq y$ and $y \leq x$, then $x = y$). We write $x<y$ if $x \leq y$ and $y \not\leq x$. If $\leq$ is a partial ordering, $x<y$ is then equivalent to $x \leq y$ and $x \neq y$.
		
		To any $x \in \stateset$, we associate the sets $\upclose x \defeq \{y \mid x \leq y\}$ and $\downclose x \defeq \{y \mid y \leq x\}$. Moreover, for $A \subseteq \stateset$, we let $\upclose A \defeq \bigcup_{x \in A} \upclose x$ and $\downclose A \defeq \bigcup_{x \in A} \downclose x$. We say that a set $A$ is \emph{upward closed} if $A = \upclose A$. Similarly, $A$ is \emph{downward closed} if $A = \downclose A$. A \emph{basis} of an upward-closed set $A$ is a set $B \subseteq \stateset$ such that $A = \upclose B$. 
		
		We say $(\stateset,\leq)$ is \emph{well-founded} if there is no infinite strictly decreasing sequence 
		of elements of $\stateset$. An \emph{antichain} is
		a subset $A \subseteq \stateset$ such that any two distinct elements in the subset are incomparable, \ie, for
		every distinct $x, y \in A$, we have $x \not\leq y$ and $y \not\leq x$. 
		For example, consider the alphabet $\alphset = \{a,b\}$. There exists an infinite antichain $\{b, ab, aab, ...\}$ with respect to the prefix ordering over $\alphset^\ast$.

		An \emph{ideal} is a downward-closed set $I \subseteq X$
		that is also \emph{directed}, \ie, it is non-empty and, for every $x, y
		\in I$, there exists $z \in I$ such that $x \leq z$ and $y \leq
		z$. The set of ideals is denoted by $\ideals(X)$.

		\smallskip
		\noindent\textbf{Well-quasi-orderings.}~ For the following definitions, we fix a \qo $(\stateset, \leq)$. When a \qo satisfies some additional property, we call it a well-quasi-ordering:
		
		\begin{defi}
			A \emph{well-quasi-ordering} (\wqo) is a \qo $(\stateset,\leq)$ such that every infinite sequence $\initstate, x_1, x_2, \ldots$ over $\stateset$ contains an \emph{increasing pair}, i.e.,
			there are $i < j$ such that $x_i \leq x_j$.
		\end{defi}
		
		For example, the set of natural numbers $\nat$, along with the standard ordering $\leq$, is a \wqo. Moreover, $(\nat^k, \leq)$, \ie, the set of vectors of $k \geq 1$ natural numbers with component-wise ordering, is a \wqo \cite{dickson_finiteness_1913}. On the other hand, the prefix ordering on words over an alphabet $\alphset$ is not a \wqo (if $|\alphset| > 1$) since it contains infinite antichains: in the infinite sequence $b, ab, a^2b, a^3b,...a^nb,...$, we have  $a^ib \not\pref a^jb$ for all $i<j$ and $i,j \in \nat$.
		
		In general, for a \qo, upward-closed sets do not necessarily have a \emph{finite} basis. However, from \cite{higman_ordering_1952}, we know that every upward-closed set in a {\wqo} has a finite basis.
		
		We have the following equivalent characterisation of \wqos.

		\begin{propC}[\cite{erdos_partition_1956,schmitz_algorithmic_2012}]\label{erdos2} The following are equivalent, given a \qo $(\stateset,\leq)$:
			\begin{enumerate}
				\item $(\stateset,\leq)$ is a \wqo.
				\item Every infinite sequence in $\stateset$ has an infinite increasing subsequence with respect to $\leq$.
				\item Every upward-closed non-empty subset in $\stateset$ has a finite basis.
				\item $(\stateset,\leq)$ is well-founded and contains no infinite antichain.
			\end{enumerate}
		
		\end{propC}

		On the other hand, downward-closed subsets of \qos enjoy the following property:
		
		\begin{propC}[\cite{erdos_partition_1956}]\label{erdos}
			A \qo $(\stateset,\leq)$ contains no infinite antichain \ifff every down\-ward-closed set decomposes into a finite union of ideals.
		
		\end{propC}
		
		 The above proposition is useful for designing the forward coverability algorithm described in Section~\ref{sec:dec-probs}.
		 
		 \smallskip
		 Next, we introduce transition systems:
		
		\noindent\textbf{Transition systems.}~
		A \emph{transition system} is a pair $\system = (\stateset,\srun)$ where $\stateset$ is the set of states and ${\srun} \subseteq \stateset \times \stateset$ is the transition relation. We write $x \srun{} y$ for $(x, y) \in {\srun}$. Moreover, we let $\srunp{*}$ be the transitive and reflexive closure of the relation~$\srun$, and $\srunp{+}$ be the transitive closure of $\srun$.
		
		Given a state $x \in \stateset$, we write $\Post_\system(x) = \{y\in \stateset \mid x \srun{} y\}$ for the set of immediate successors of $x$. Similarly, $\Pre_\system(x) = \{y\in \stateset \mid y \srun{} x\}$ denotes the set of its immediate predecessors.
		
		We call $\system$ \emph{finitely branching} if, for all $x \in \stateset$, the set $\Post_\system(x)$ is finite. The \emph{reachability set} of $\system$ from $x \in \stateset$ is defined as $\Post_\system^*(x) = \{y \in \stateset \mid x \srunp{*} y\}$. Note that when $\system$ is clear from the context, we may drop the subscript and write, e.g., $\Post^*(x)$. We say that a state $y$ is reachable from $x$ if $y \in \Post^*(x)$.  
		
		We recall that the \emph{reachability tree} from an initial state $\initstate$ in $\system$ is a tree with a root node labelled by $\initstate$. Then, for all $y$ such that $\initstate \srun y$, we add a vertex labelled with $y$ and add an edge from the root node to the node labelled with $y$. We then compute $\Post(y)$ and once again add vertices labelled with each state in $\Post(y)$. We repeat this for every vertex along the tree. Note that we can have multiple vertices labelled with the same state, and moreover, the reachability tree can be infinite even if the reachability set from the initial state is not.
		
		A \emph{(well-)ordered transition system} is a triple $\system = (\stateset, \srun, \leq)$ consisting of a transition system $( \stateset,\srun)$ equipped with a \qo (resp., \wqo) $(\stateset, \leq)$.
		An ordered transition system $\system = (\stateset, \srun, \leq)$ is \emph{effective} if $\leq$ and $ \srun$ are decidable. We say that a state $y$ is \emph{coverable} from $x$ if $y \in \downclose \Post^*(x)$.

		\begin{defiC}[\cite{finkel_reduction_1990}]
			A \emph{well-structured transition system (\WSTS)} is a well-ordered transition system $\system = (\stateset, \srun, \leq)$ that satisfies (general) \emph{monotony}: for all $x, y, x' \in \stateset$, we have:
			$
			x \leq y \land x \srun x' \implies \exists y' \in \stateset\textup{: } x' \leq y' \land y \srunp{*} y'. 
			$
		\end{defiC}
		
		We define other types of monotony. We say that a well-ordered transition system $\system = (\stateset, \srun, \leq)$ satisfies  \emph{strong monotony} (resp., \emph{transitive monotony}) if, for all $x,y,x' \in \stateset$ such that $x \leq y$ and $x \srun x'$, there is $y' \in \stateset$ such that $x' \leq y'$ and $y \srun y'$ (resp., $y \srunp{+} y'$). The transition system $\system$ satisfies \emph{strict monotony} if, for all $x,y,x' \in \stateset$ such that $x < y$ and $x \srun x'$, there is $y' \in \stateset$ such that $x' < y'$ and $y \srun y'$.\\
		
		\subsection{Decision problems for transition systems.} \label{subsec:decprobs} \label{sec:dec-probs} We recall the following well-known decision problems.

		\begin{defi}\label{def:verification-problems}
			Given an ordered transition system $\system = (\stateset, \srun, \leq)$ and an initial state $\initstate \in \stateset$:
			\begin{itemize}
				\item \emph{The non-termination problem}: 
				Is there an infinite sequence of states $x_1, x_2,\ldots$ such that $\initstate \srun x_1  \srun x_2 \srun \ldots$ ?
				
				\item \emph{The boundedness problem}: Is $\Post^*_\system(\initstate)$ finite?
				
				\item  \emph{The coverability problem}: Given a state $y \in \stateset$, is $y$ coverable from $\initstate$? 
				
			\end{itemize}
			
		\end{defi}
		
		It is folklore \cite{finkel_reduction_1990,finkel_well-structured_2001} that non-termination is decidable for finitely branching \WSTS with transitive monotony, and that boundedness is decidable for finitely 
		branching \WSTS $\system = (\stateset, \srun, \leq)$ where  $\leq$ is a partial ordering and $ \srun$ is strictly monotone. In both cases, we suppose that the \WSTS are effective and that $\Post(x)$ is computable for all $x \in \stateset$. 
		
		\smallskip
		\noindent\textbf{Coverability problem.}~ We recall that coverability is decidable for a large class of \WSTS:

		\begin{thmC}[\cite{finkel_well-structured_2001,abdulla_algorithmic_2000}]\label{thm:finitebasiscov}
			The coverability problem is decidable for effective \WSTS $\system = (\stateset, \srun, \leq)$ equipped with an algorithm that, for all finite subsets $I \subseteq \stateset$, computes a finite basis $\pb(I)$ of $\upclose \Pre(\upclose I)$.
		\end{thmC}
		
		Assume $\system = (\stateset, \srun, \leq)$ is a \WSTS and $x \in \stateset$ is a state. The \emph{backward coverability algorithm} involves computing  (a finite basis of) $\Pre^*(\upclose x)$ as the limit of the infinite increasing sequence $\upclose I_0 \subseteq \upclose I_1 \subseteq \ldots$ where $I_0= \{x\}$ and $I_{n+1} \defeq I_n \cup \pb(I_n)$. Since there exists an integer $k$ such that $\upclose I_{k+1} = \upclose I_k$, the finite set $I_k$ is computable (one may test, for all $n$, whether $\upclose I_{n+1} = \upclose I_n$) and $I_k$ is then a finite basis of  $\Pre^*(\upclose x)$ so one deduces that coverability is decidable.
		
		Coverability can also be decided using the \emph{forward coverability algorithm} that relies on two semi-decision procedures (as described below). It relies on Proposition~\ref{erdos} and enumerates finite unions of ideals composing inductive invariants. It shows that the ``no infinite antichain" property is sufficient to decide coverability; hence, the \wqo hypothesis is not necessary. It applies to the class of \emph{well-behaved transition systems}, which are more general than \WSTSs.  A well-behaved transition system (\WBTS) is an ordered transition system $\system = (\stateset, \srun, \leq)$ with monotony such that $(\stateset, \leq)$
		contains no infinite antichain. We describe effectiveness hypotheses that allow
		manipulating downward-closed sets in \WBTSs.

		\begin{defiC}[\cite{blondin_well_2017}]
			A class $C$ of \WBTSs is \emph{ideally effective} if, given
			${\system = (\stateset, \srun{}, \leq) \in C}$,
			\begin{itemize}
				\item the set of encodings of $\ideals(\stateset)$ is recursive,
				\item the function mapping the encoding of a state $x \in \stateset$ to the
				encoding of the ideal $\downclose{x} \in \ideals(\stateset)$ is computable;
				\item inclusion of ideals of $\stateset$ is decidable;
				\item the downward closure $\downclose{\Post(I)}$ expressed as a finite
				union of ideals is computable from the ideal $I \in \ideals(\stateset)$.
			\end{itemize}
		\end{defiC}

		\begin{thmC}[\cite{blondin_well_2017}]
			The coverability problem is decidable for ideally effective \WBTS.
		\end{thmC}
		
		The result is derived from the design of two semi-decision procedures  where downward-closed sets are represented by their finite decomposition in ideals, and this is effective. Procedure \ref{proc:cov} checks for coverability of $y$ from $\initstate$, by recursively computing $\downclose \initstate$, $\downclose(\downclose \initstate \cup \Post(\downclose \initstate))$ and so on. This procedure terminates only if $y$ belongs to one of these sets. Hence, it terminates if $y$ is coverable.

		\begin{procedure}[!ht]
			\caption{: Checks for coverability}
			\textbf{input:}  $\system = (\stateset, \srun, \leq)$ and $\initstate,y$
			\begin{algorithmic} 
				\STATE $D  := \downclose \initstate$
				\WHILE{ $y \notin D$ }
				\STATE	 $D := \downclose (D \cup \Post_\system(D))$
				\ENDWHILE
				\RETURN  ``$y$ is coverable from $\initstate$"
			\end{algorithmic} \label{proc:cov}
		\end{procedure}
		
		Hence, we deduce:
		
		\begin{propC}[\cite{blondin_well_2017}]
			For an ideally effective \WBTS $\system = (\stateset,\srun, \leq)$, an initial state $\initstate$, and a state $y$, Procedure~1 terminates \ifff $y$ is coverable from $\initstate$.
		\end{propC}
		
		Procedure \ref{proc:noncov} enumerates all downward-closed subsets (by means of their finite decomposition in ideals) in some fixed order $D_1, D_2, \ldots$ We remark that this enumeration is effective since $\system$ is ideally effective. Furthermore, it checks for all $i$, $D_i \subseteq \stateset$ and $\downclose \Post(D_i) \subseteq D_i$, and if such a set $D_i$ contains $\initstate$ but not $y$. If it does contain $\initstate$, it is an over-approximation of $\Post^*( \initstate)$. Hence, if there is such a set $D_i$ with $\initstate \in D_i$ but $y \notin D_i$, it is a certificate of non-coverability. Moreover, this procedure terminates if $y$ is non-coverable because $\downclose \Post^*(\initstate)$ is such a set, and hence, will eventually be found.
		
		\begin{procedure}[!ht]
			\caption{: Checks for non-coverability}
			\textbf{input:}  $\system = (\stateset, \srun, \leq)$ and $\initstate,y$
			\begin{algorithmic} 
				\STATE$\textbf{enumerate} \text{ downward-closed sets } D_1, D_2, \ldots$ 
				\STATE $i : = 1$
				\WHILE{ $\neg (\downclose \Post(D_i) \subseteq D_i \text{ } \AND \text{ } \initstate \in D_i \text{ } \AND  \text{ } y \notin D_i)$}
				\STATE	 $i := i+1$
				\ENDWHILE
				\RETURN\FALSE
			\end{algorithmic} \label{proc:noncov}
		\end{procedure}  
		
		Therefore, we have:
		
		\begin{propC}[\cite{blondin_well_2017}]
			For a \WBTS $\system = (\stateset,\srun, \leq)$, an initial state $\initstate$ and a state $y$, Procedure 2 terminates \ifff $y$ is not coverable from $\initstate$.
		\end{propC}
		
		\subsection{Labelled transition systems} 
		Next, we define labelled transition systems, which are transition systems where the transitions are equipped with labels.
		
			A \emph{labelled transition system (\LTS)} is a tuple $\system = {(\stateset, \labelset, \srun, \initstate)}$ where $\stateset$ is the set of states, $\labelset$ is the finite action alphabet, ${\srun} \subseteq {\stateset \times \labelset \times \stateset}$ is the transition relation, and $\initstate \in \stateset$ is the initial state.
		
		\begin{defi}
			A \emph{quasi-ordered labelled transition system (\OLTS)} is defined as a tuple $\system = (\stateset, \labelset, \srun, \leq, \initstate)$ where $(\stateset, \labelset, \srun, \initstate)$ is an \LTS and $(\stateset, \leq)$ is a \qo.
		\end{defi}
		
		In the case of an \LTS or \OLTS, we write $x \srunp{a} x'$ instead of $(x,a,x') \in {\to}$.
		For $\sigma \in \labelset^\ast$,
		$x \srunp{\sigma} x'$ is defined as expected.
		We also let
		$x \srun x'$ if $(x,a,x') \in {\to}$ for some $a \in \labelset$,
		with closures $\srunp{*}$ and $\srunp{+}$. We let $\Traces{\system}=\{ w \in \labelset \mid \initstate \srunp{w} x \text{ for some } x \in \stateset \}$ be the set of traces.
		
		We call an \OLTS $\system$ \emph{effective} if $\leq$ and, for all $a \in \labelset$, $\srunp{a}$ are decidable.

		\begin{rem} We can similarly define a labelled WSTS as an \OLTS such that the ordering is well and it satisfies the general monotony condition (canonically adapted to take care of the transition labels).

			Moreover, we lift the decision problems from Definition~\ref{def:verification-problems} to \OLTS in the obvious way. 
		\end{rem}
		
		\subsection{Classes of labelled transition systems}
		
		In this paper, we mainly study subclasses of two well-known classes of transition systems, namely counter machines and FIFO machines. Both these classes are known to simulate Turing machines and, hence, have undecidable boundedness and non-termination.
		
		\smallskip
		\noindent\textbf{Counter machines.}~ Counter machines, also known as Minsky machines, are finite-state machines that manipulate counters, which are variables that store non-negative integers. Transitions of a counter machine, besides changing control-states, may also perform a specified operation on a counter: increment by one or decrement by one, along with a set of counters to be tested for zero. We formally define a counter machine below.
		
		\begin{defi}
			A \emph{counter machine} (with zero tests) is a tuple
			$\countermachine =(Q,\counterset , T, q_0)$.
			Here, $Q$ is the finite set of \emph{control-states}
			and $q_0 \in Q$ is the \emph{initial control-state}.
			Moreover, $\counterset$ is a finite set of \emph{counters}
			and $T \subseteq Q \times \Action_\countermachine \times Q$
			is the transition relation where $\Action_\countermachine =
			\{\inc{\counter},\dec{\counter}, \noop \mid \counter \in \counterset\} \times 2^\counterset$
			(an element of $2^\counterset$ will indicate the set of counters to be tested to $0$).
		\end{defi}
		
		The counter machine $\countermachine$ induces an LTS
		$\system_\countermachine = (X_\countermachine,\Action_\countermachine,\srun_\countermachine, x_0)$
		with set of states $X_\countermachine = Q \times \mathbb{N}^\counterset$.
		In $(q,\counterval) \in X_\countermachine$, the first component $q$ is the current control-state and
		$\counterval = (\counterval_{\counter})_{\counter \in \counterset}$
		represents the counter values.
		The initial state is then $x_0 = (q_0, \counterval_0)$ with $\counterval_0 = (0, 0, \ldots, 0)$. 
		For $op \in \{\mathsf{inc},\mathsf{dec}\}$, $\counter \in \counterset$, and $Z \subseteq \counterset$ ($Z$ is the set of counters tested for zero), there is a transition
		$
		(q,\counterval) \srunp{op(\counter), Z}_\countermachine (q',\countervalp)
		$
		if $(q, (op(\counter), Z), q') \in T$,
		$\counterval_{\counterb} = 0$ for all $\counterb \in Z$,
		$\countervalp_{\counter} = \counterval_{\counter} + 1$ if $op = \mathsf{inc}$ and
		$\countervalp_{\counter} = \counterval_{\counter} - 1$ if $op = \mathsf{dec}$, and
		$\countervalp_{\counterb} = \counterval_{\counterb}$ for all $\counterb \in \counterset \setminus 
		\{\counter\}$. 
		
		For $op = \noop$, and $Z \subseteq \counterset$, there is a transition $
		(q,\counterval) \srunp{op, Z}_\countermachine (q',\countervalp)
		$
		if $(q, (op, Z), q') \in T$,
		$\counterval_{\counterb} = 0$ for all $\counterb \in Z$ (applies the zero tests), and
		$\countervalp_{\counter} = \counterval_{\counter}$ for all $\counter \in \counterset$. We sometimes omit writing $\noop$ and label the transition with only the set of counters to be tested to zero, or we write $\textsf{zero}(Z)$. Similarly, we omit $Z$ if $Z = \emptyset$.
		
		\smallskip
		\noindent\textbf{FIFO machines.}~ Next, we study FIFO machines. FIFO machines can be viewed as finite state machines equipped with one or more (potentially unbounded) channels where messages can be sent to and received from. \cite{brand_communicating_1983} showed that general FIFO machines can simulate Turing machines, and hence, it is undecidable to verify boundedness. However, it is still a widely used model to represent asynchronously communicating systems.
		
		\begin{defi}\label{def:fifo}
			A FIFO machine $\machine$ over the set of channels $\chanset$ is a tuple $\machine = (Q, \alphabetset, T, q_0)$ where $Q$ is a finite set of control-states, $\alphabetset$ is the finite message alphabet, and $q_0 \in Q$ is an initial control-state. Moreover, $T \subseteq Q \times \chanset \times \{!,?\}\times \alphabetset \times Q$ is the transition relation, where $\chanset \times \{!\} \times \alphabetset$ and $\chanset \times \{?\} \times \alphabetset$ are the set of send and receive actions, respectively.
		\end{defi}
		
		The FIFO machine $\machine$ induces the LTS $\system_\machine = (X_\machine, \labelset_\machine, \srun_\machine, x_0)$. Its set of states is $X_\machine = Q \times (\alphabetset^*)^{\chanset}$. In $(q,\chcontents) \in X_\machine$, the first component $q$ denotes the current control-state, and $\chcontents = (\chcontents_\ch)_{\ch \in \chanset}$ denotes the contents $\chcontents_c \in \alphabetset^*$ for every channel $\ch \in \chanset$. The initial state is $x_0 = (q_0, \chempty)$, where $\chempty$ denotes that every channel is empty. Moreover, $\labelset_\machine = \chanset \times \{!, ?\} \times \alphabetset$. The transitions are given as follows:\begin{itemize}
			\item $(q, \chcontents) \srunp{\ch!a}_\machine (q', \chcontentsp)$ if $(q, \ch!a, q') \in T$, $\chcontentsp_\ch = \chcontents_\ch \cdot a$, and $\chcontentsp_\chp = \chcontents_\chp$ for all $\chp \in \chanset \setminus \{\ch\}$.
			\item $(q, w) \srunp{\ch?a}_\machine (q', w')$ if $(q, \ch?a, q') \in T$, $\chcontents_\ch = a \cdot \chcontentsp_\ch$, and $\chcontentsp_\chp = \chcontents_\chp$ for all $\chp \in \chanset \setminus \{\ch\}$.
		\end{itemize}
		The index $\machine$ may be omitted whenever $\machine$ is clear from the context.
		When there is no ambiguity, we refer to machines and their associated LTS interchangeably.
		
		Note that, in general, FIFO machines with a single channel are as powerful as FIFO machines with multiple channels. When we denote FIFO machines with a single channel, we omit specifying the set $\chanset$, and hence, we denote the set of labels as $\labelset_\machine = \{!, ?\} \times \alphabetset$.
		
		For FIFO machines, we define the extended prefix ordering, denoted by $\extpref$ as follows:  $(q,\chcontents) \extpref (q',\chcontentsp)$ if $q = q'$ and for all $\ch \in \chanset$, $\chcontents_\ch$ is a prefix of $\chcontentsp_\ch$, i.e., $\chcontents_\ch \pref \chcontentsp_\ch$.

		\section{Branch-well-structured transition systems}\label{sec:termbound}

		In this section, we generalise \wqo and monotony such that these properties only need to hold for states along a branch in the reachability tree. To define these notions, we use labels on the transitions; hence, we consider labelled transition systems.
	
		\subsection{Branch-\wqo}\label{subsec:branch-wqo}
		Consider an \OLTS $\system = (X, \labelset, \srun, \leq, \initstate)$.
		A \emph{run} (or \emph{branch}) of $\system$ is a finite or infinite sequence $\rho=(\initstate \srun x_1)(x_1 \srun x_2)...$ simply written  $\rho=\initstate \srun x_1 \srun x_2 \ldots$. We denote the set of states visited along $\rho$ as $\Rho{\rho} = \{\initstate,x_1,x_2,\ldots\}$. We say that $\rho$ is \emph{branch-\wqo} if $\Rho{\rho}$ is \wqo w.r.t.~$\leq$. 
		
		\begin{defi}
			An \OLTS $\system = (X, \labelset, \srun, \leq, \initstate)$ is \emph{branch-\wqo} if every run of $\system$ is branch-\wqo.
		\end{defi}

		\begin{exa}\label{ex:fifo-machine-1}
			Consider the FIFO machine $\machine_1$ in Figure~\ref{fig:fifo-machine-1} with one FIFO channel. In control-state $q_0$, it loops by sending letter $a$ to the channel. Then, we may go, non-deterministically, to control-state $q_1$ by sending letter $b$ once, and then we stop.
			Let us consider the set of states $X_1 = \{q_0,q_1\} \times \{a,b\}^*$ together with the extended prefix ordering $\extpref$, i.e. $(q,u) \extpref (q',u')$ if $q = q'$ and $u \pref u'$.  The reachability set of $\machine_1$ from $(q_0,\varepsilon)$ is equal to $\Post^*(q_0,\varepsilon) = \{(q_0, w) \mid w \in a^*\} \cup \{ (q_1, w') \mid  w' \in a^*b\}$. 
			Note that $\extpref$ is not a \wqo since elements of the set $\{(q_1, w') \mid w' \in a^*b\}$ form 
			an infinite antichain for $\extpref$. However, the reachability tree of $\machine_1$ is branch-\wqo  for the initial state $(q_0, \varepsilon)$ (every branch is either finite or \wqo). Hence, there exist branch-\wqo \OLTS $\system = (\stateset, \labelset, \srun, \leq, \initstate)$ such that $(\stateset,\leq)$ is not a \wqo.
			
		\end{exa}
		
		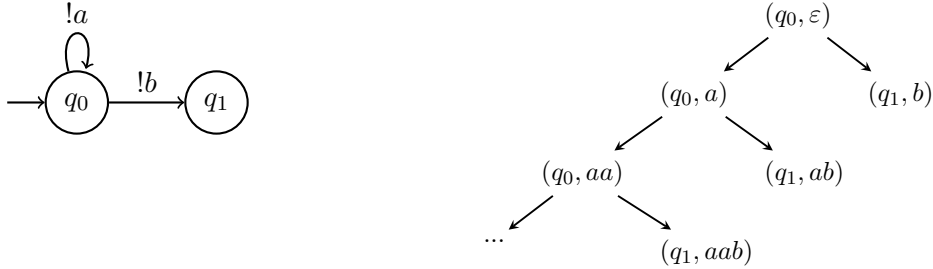
\begin{figure*}
			\begin{multicols}{2}
				\centering
				\begin{tikzpicture}[->, auto, thick]
					\node[place, initial, initial text=] (p1) {$q_0$};
					\node[place] (p2) [right= of p1]{$q_1$};
					
					\path[->]
					
					(p1) edge [loop above] node[swap] {$!a$} (p1)
					(p1) edge node[] {$!b$} (p2);
				\end{tikzpicture}\par
				\centering
				\scalebox{0.9}{
					\begin{tikzpicture}[->, auto, thick]
						\node[] (p1) {$(q_0,\varepsilon)$};
						\node[] (p2) [below right= 0.5cm and 0.25cm of p1]{$(q_1, b)$};
						\node[] (p3) [below left= 0.5cm and 0.25cm of p1]{$(q_0, a)$};
						\node[] (p4) [below left= 0.5cm and 0.25cm of p3]{$(q_0, aa)$};
						\node[] (p5) [below right= 0.5cm and 0.25cm of p3]{$(q_1, ab)$};
						\node[] (p6) [below left= 0.5cm and 0.25cm of p4]{$...$};
						\node[] (p7) [below right= 0.5cm and 0.25cm of p4]{$(q_1, aab)$};

						\path[-stealth]
						
						(p1) edge node[] {} (p2)
						(p1) edge node[] {} (p3)
						(p3) edge node[] {} (p4)
						(p3) edge node[] {} (p5)
						(p4) edge node[] {} (p6)
						(p4) edge node[] {} (p7)
						;
						
				\end{tikzpicture}}
			\end{multicols}
			\caption{The FIFO machine $\machine_1$ (left), and the corresponding (incomplete) infinite reachability tree (right) with initial state $(q_0, \varepsilon)$. We see that the induced transition system is branch-\wqo. \label{fig:fifo-machine-1}} 
		\end{figure*}

		However, unlike \wqo, the notion of branch-\wqo depends on the initial state considered. We will see below that there exists a system $\system = (X, \Sigma, \srun, \leq, x_0)$ and $x_0' \in X$  such that $\system$ is branch-\wqo but $(X, \Sigma, \srun, \leq, x_0')$ is not branch-\wqo. 
		
		\begin{figure*}[b]
			\begin{multicols}{2}
				\centering
				\begin{tikzpicture}[->, auto, thick]
					\node[place] (p1) {$q_0$};
					\node[place] (p2) [right= of p1]{$q_1$};
					\node[place] (p3) [right= of p2]{$q_2$};
					\path[->]
					
					(p1) edge [loop above] node[swap] {$!a$} (p1)
					(p1) edge node[] {$!b$} (p2)
					(p2) edge [bend left] node[] {$?c$} (p3)
					(p3) edge [bend left] node[] {$!b$} (p2)
					(p3) edge [loop right] node[swap] {$?b$} (p3)
					(p2) edge [loop above] node[swap] {$?c$} (p2)
					(p3) edge [loop above] node[swap] {$!c$} (p3);
				\end{tikzpicture}\par
				\centering
				\scalebox{0.9}{
					\begin{tikzpicture}[->, auto, thick]
						\node[] (p1) {$(q_2,\varepsilon)$};
						\node[] (p2) [below = 0.33cm of p1]{$(q_2, c)$};
						\node[] (p4) [below right= 0.5cm and 0.25cm of p2]{$(q_2, cc)$};
						\node[] (p5) [below left= 0.5cm and 0.25cm of p2]{$(q_1, cb)$};
						\node[] (p6) [below right= 0.5cm and 0.25cm of p4]{$...$};
						\node[] (p7) [below = 0.33cm of p4]{$(q_1, ccb)$};
						\node[] (p8) [below left= 0.5cm and 0.25cm of p5]{$(q_1, b)$};
						\node[] (p9) [below = 0.33cm of p5]{$(q_2, b)$};				
						\node[] (p10) [below right= 0.5cm and 0.25cm of p9]{$...$};
						\node[] (p11) [below left= 0.5cm and 0.25cm of p9]{$...$};
						\node[] (p12) [below right= 0.5cm and 0.25cm of p7]{$...$};
						\node[] (p13) [below left= 0.5cm and 0.25cm of p7]{$...$};
						\node[] (p14) [below = 0.7cm of p9]{$(q_2, \varepsilon)$};
						\node[] (p15) [below = 0.7cm of p14]{$(q_2, c)$};
						\node[] (p16) [below = 0.7cm of p15]{$(q_2, cc)$};
						\node[] (p17) [below = 0.7cm of p16]{$(q_1, ccb)$};
						\node[] (p18) [below left= 0.5cm and 0.25cm of p14]{$...$};
						\node[] (p19) [below left= 0.5cm and 0.25cm of p15]{$...$};
						\node[] (p20) [below left= 0.5cm and 0.25cm of p16]{$...$};
						\node[] (p21) [below left= 0.5cm and 0.25cm of p17]{$...$};
						\node[] (p22) [below = 0.7cm of p17]{$...$};

						\path[-stealth]
						
						(p1) edge node[] {} (p2)
						(p2) edge node[] {} (p4)
						(p2) edge node[] {} (p5)
						(p4) edge node[] {} (p6)
						(p4) edge node[] {} (p7)
						(p5) edge node[] {} (p8)
						(p5) edge node[] {} (p9)
						(p9) edge node[] {} (p10)
						(p9) edge node[] {} (p11)
						(p9) edge node[] {} (p14)
						(p7) edge node[] {} (p12)
						(p7) edge node[] {} (p13)
						(p14) edge node[] {} (p15)
						(p14) edge node[] {} (p18)
						(p15) edge node[] {} (p16)
						(p15) edge node[] {} (p19)
						(p16) edge node[] {} (p17)
						(p16) edge node[] {} (p20)
						(p17) edge node[] {} (p21)
						(p17) edge node[] {} (p22)
						;
						
				\end{tikzpicture}}
			\end{multicols}
			\caption{The FIFO machine $\machine_2$ (left) and the incomplete reachability tree from $(q_2, \varepsilon)$ (right).\label{fig:branch-init}} 
		\end{figure*}
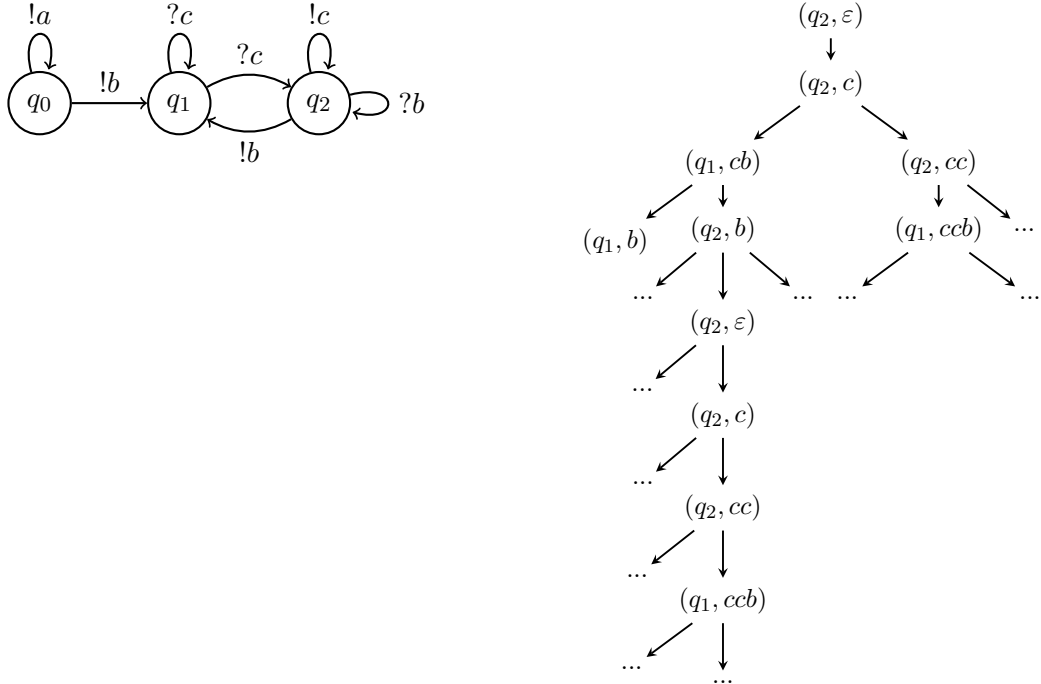
		
			\begin{exa}\label{ex:branch-init}
			Consider the FIFO machine $\machine_2$ in Figure~\ref{fig:branch-init} with one FIFO channel. If we start from $(q_0, \varepsilon)$, it behaves exactly as the FIFO machine $\machine_1$. However, if we change the initial state to $(q_2, \varepsilon)$, then it could either loop by sending $c$, or non-deterministically send a $b$ and go to control-state $q_1$. Then, if at least one $c$ has been sent, it can either loop again receiving the letters $c$, or once again non-deterministically come back to $q_2$ upon receiving $c$. 		
			There exists an infinite run $\rho$ with the prefix: $(q_2, \varepsilon) \srunp{!c} (q_2, c) \srunp{!b} (q_1, cb) \srunp{?c} (q_2, b) \srunp{?b} (q_2, \varepsilon) \srunp{!c. !c. !b} (q_1, ccb) \srun ...$ such that all the elements of the infinite set $B = \{{(q_1, cb)}, {(q_1, ccb)},\allowbreak \ldots, \allowbreak {(q_1, c^nb)}, \ldots\}$ are visited along $\rho$, i.e $B \subseteq \stateset_\rho$. Hence, we have an infinite sequence of incomparable elements along $\rho$, i.e. $\stateset_\rho$ is not \wqo, and therefore, the system is not branch-\wqo.

		\end{exa}
		
		We now show that branch-\wqos enjoy some of the good properties of \wqos.
		Let $\system_1 = (X_1, \Sigma, \srun_1, \leq_1, x_{0,1})$ and $\system_2 = (X_2, \Sigma, \srun_2, \leq_2, x_{0,2})$ be two branch-\wqos. We consider their product to be $\system = (X, \Sigma, \srun, \leq, x_0)$, where $X = X_1 \times X_2$, and $((x_1, x_2), a, (x'_1, x'_2)) \in \srun$ if $(x_1, a, x'_1) \in \srun_1$ and $(x_2, a, x'_2) \in 
				\srun_2$. Moreover, we consider the ordering to be component-wise, i.e. we have $(x_1,x_2) \leq (x'_1, x'_2)$ if $x_1 \leq_1 x'_1$ and $x_2 \leq_2 x'_2$; finally $x_0=(x_{0,1},x_{0,2})$.

		\begin{prop} 
				The product of finitely many branch-\wqos is branch-\wqo.

		\end{prop}
		\begin{proof}
		
				We show this for two branch-\wqos, but this proof can be extended to finitely many branch-\wqos.

	Let us consider an infinite run $\rho$ in the product branch-\wqo $\system$. If we consider the set $\stateset_\rho$ along with the order $\leq_1$, we can extract an infinite subsequence $B  = (x_n, y_n)_{n \in \nat} \subseteq \stateset_\rho$ such that $(x_n)_{n \in \nat}$ forms an  increasing sequence wrt. $\leq_1$ (since $(\stateset_1, \leq_1)$ is branch-\wqo). Now,  within this subsequence $B$, there exists an increasing subsequence $(y_{n_m})_{m \in \nat}$ (since the branch is also \wqo wrt. $\leq_2$). Hence, we have an infinitely long sequence $(x_{n_{m_k}}, y_{n_{m_k}})_{k \in \nat}$.
			
		\end{proof}
		
		We can similarly prove that the same holds true for disjoint unions and intersection of branch-\wqos.

		\subsection{Branch-monotony}\label{subsec:branch-monotony}
		Now that we have relaxed the notion of branch-\wqo, we shall look at a generalisation of strong monotony, which we will refer to as branch-monotony.

		\begin{defi}
			An \OLTS $\system = (X, \Sigma,$ $\srun, \leq, x_0)$
			is \emph{branch-monotone} if, for all $x, x' \in X$, $\sigma \in \Sigma^*$ such that $x \srunp{\sigma} x'$ and $x \leq x'$, there exists a state $y$ such that $x' \srunp{\sigma} y$ and $x' \leq y$.
		\end{defi}
		
		\begin{prop}
			Let $\system$ be a branch-monotone \OLTS and let there be states $x,x'$ such that $x \srunp{\sigma} x'$  and $x \leq x'$, with $\sigma \in \Sigma^*$.
			Then, for any $n \geq 1$, there exists $y_n \in X$ such that $x \srunp{\sigma^n} y_n$ with $x \leq y_n$.
		\end{prop}
		\begin{proof}
			We prove this by induction on $n$. Base case: Let $n=1$. From the hypothesis, we immediately have a run $x \xrightarrow{\sigma} x'$ such that $x \leq x'$. Hence, the property holds trivially for the base case, with $y_1=x'$.
			
			Let us assume that it holds for $n$. We will now show that it also holds for $n+1$. Let $x \xrightarrow{\sigma} x'$  be a finite run satisfying $x \leq x'$. Furthermore, from the induction hypothesis, there exists a state $y_n$ such that $x \xrightarrow{\sigma} x' \xrightarrow{\sigma^{n-1}}y_n$ such that $x \leq y_n$. Moreover, by the definition of branch-monotony, there exists a state $y$ such that $ x' \xrightarrow{\sigma} y$ and $x' \leq y$. Furthermore, if we consider the run $x' \srunp{\sigma} y$, again by induction hypothesis, there exists a state $y'_n$ such that $x' \srunp{\sigma^n} y'_n$ and $x \leq y'_n$. Hence, there exists a run $x\xrightarrow{\sigma} x' \srunp{\sigma^n} y'_{n}$, and $x \leq x' \leq y'_n$, so by transitivity, $x \leq y'_n$. In other words, we have $x \srunp{\sigma^{n+1}} y'_n = y_{n+1}$, and this completes our proof.
		\end{proof}

		As in the case of general monotony, \emph{strict} branch-monotony is defined using strict inequalities in both cases. 
		
		\begin{exa}
			Consider $\machine_1$ from Figure~\ref{fig:fifo-machine-1} once again. Note $\machine_1$ induces an \OLTS
			by considering the actions on the edges to be the labels. Moreover, $\machine_1$ is branch-monotone.
			For every $x \srunp{\sigma} x'$ such that $x \extpref x'$ and $\sigma \in \Sigma^*$, the only case is that $x = (q_0, a^n)$, $x' = (q_0, a^{n+k})$, for some $n,k \in \mathbb{N}$. Moreover, $\sigma \in (!a)^*$. Hence, we can always repeat $\sigma$ from $x'$ such that $x' \srunp{\sigma} y = (q_0, a^{n+k+k})$. Therefore, $x' \extpref y$. We deduce that $\machine_1$ is branch-monotone.
		\end{exa}

		\subsection{Branch-WSTS}\label{subsec:branch-wsts} We are now ready to extend the definition of WSTS.
		
		\begin{defi}
			A \emph{branch-WSTS} is an \OLTS $\system = (X, \Sigma,$ $\srun, \leq, x_0)$ that is
			finitely branching, branch-monotone, and branch-\wqo.
		\end{defi}
		
		When we say, without ambiguity, that a machine $\machine$ is branch-\wqo, WSTS, or branch-WSTS, we mean that the ordered transition system $\system_{\machine}$, induced by machine $\machine$, is branch-\wqo, WSTS, or branch-WSTS, respectively. We will explicitly define the transition system induced by FIFO machines in Section~\ref{subsec:FIFO}.

		\begin{rem}
			
			Branch-WSTS is a strict superclass of labelled WSTS.
			For example, machine $\machine_1$ (seen in Figure~\ref{fig:fifo-machine-1}) with initial state $(q_0, \varepsilon)$ is branch-WSTS for the ordering $\extpref$ but it is not WSTS for $\extpref$ since $\extpref$ is not a \wqo on $\{q_0,q_1\} \times \{a,b\}^*$ or on the subset $\{(q_1,w) \mid w \in a^*b\}$. 
		\end{rem}

		Let us recall the \emph{Reduced Reachability Tree ($\FRT$)}, which was defined as Finite Reachability Tree in \cite{finkel_reduction_1990,finkel_well-structured_2001}.
		Suppose that $\system = (X, \Sigma, \srun, \leq, x_0)$ is an \OLTS. Then, the \emph{Reduced Reachability Tree} from $x_0$, denoted by $\FRT( \system, x_0)$, is a tree with root $n_0$ labelled $x_0$ that is defined and built as follows. For every $x \in \Post(x_0)$, we add a child vertex labelled $x$ to $n_0$. The tree is then built in the following way. We pick an unmarked vertex $c$ which is labelled with $x$: \begin{itemize}
			\item if $n$ has an ancestor $n'$ labelled with $x'$ such that $x' \leq x$, we mark the vertex $n$ as \emph{dead}, and say $n'$ \emph{subsumes} $n$.
			\item otherwise, we mark $n$ as \emph{live}, and for every $y \in \Post(x)$, we add a child labelled $y$ to $n$.
		\end{itemize}

		\begin{figure}[t]
			\centering
			\begin{tikzpicture}[->, auto, thick]
				\node[] (p1) {$(q_0,\varepsilon)$};
				\node[] (p2) [below right= 0.5cm and 0.5cm of p1]{$(q_1, b)$};
				\node[] (p3) [below left= 0.5cm and 0.5cm of p1]{$(q_0, a)$};
				\node[] (p4) [below = 0.02cm of p3]{\text{dead}};

				\path[-stealth]
				
				(p1) edge node[] {} (p2)
				(p1) edge node[] {} (p3)
				;
				
			\end{tikzpicture}
			\caption{The Reduced Reachability Tree of $\machine_1$ from $(q_0, \epsilon)$. Note that $(q_0, a)$ is dead because it is subsumed by state $(q_0, \varepsilon)$. As a matter of fact, we have $(q_0, \varepsilon) \srunp{*} (q_0, a)$ and $(q_0, \varepsilon) \extpref (q_0, a)$. State $(q_1, b)$ is also dead but it is not subsumed.} \label{fig:frt-1}
		\end{figure}
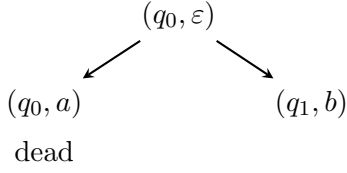

		\begin{restatable}{prop}{finite}\label{prop:finite}
			Let $\system = (X, \Sigma, \srun, \leq, x_0)$ be an \OLTS that is finitely branching and branch-\wqo.
			Then, $\FRT(\system,x_0)$ is finite.
		\end{restatable}
		\begin{proof}
			Let us assume to the contrary that $\FRT(\system,x_0)$ is infinite. 
			Since $\system$ is finitely branching and branch-\wqo, there is an  infinite branch in the reachability tree. Since this branch is branch-\wqo, there exists a pair of states, say $x_2$ and $x_3$ visited along the branch, such that $x_0 \srunp{*} x_1 \srunp{+} x_2$ and $x_1 \leq x_2$. However, since the run is infinite, there exists $x_3$ such that $x_2 \srunp{+} x_3$. We can represent the beginning of this branch with the finite prefix $n_0(x_0)\srunp{*} n_1(x_1) \srunp{+} n_2(x_2)  \srunp{+} n_3(x_3)$, in  $\FRT(\system,x_0)$,
			such that nodes $n_1, n_2, n_3$ are all different. However, since $x_1 \leq x_2$, the node $n_2(x_2)$ has been marked as dead, and the tree need not be explored further. Thus, there is a contradiction. Hence, $\FRT(\system,x_0)$ is finite. 
		\end{proof}
		
		\begin{restatable}{prop}{boundedFRT}\label{prop:boundedFRT}
			Let $\system = (X, \Sigma,$ $\srun, \leq, x_0)$ be a branch-WSTS, equipped with strict branch-monotony and such that $\leq$ is a partial ordering. The reachability set $\Post^*_\system(x_0)$ is infinite
			\ifff there exists a branch	
			$n_0(x_0)\srunp{*} n_1(x_1) \srunp{+} n_2(x_2)$ in $\FRT(\system,x_0)$ such that $x_1 < x_2$.
		\end{restatable}
		\begin{proof}
			The following proof is an adaptation of the proof of boundedness for WSTS in \cite{finkel_well-structured_2001} to branch-WSTS.

			First, let us assume $\system$ is unbounded, i.e. $\Post_{\system}^*(x_0)$ is infinite. We will show that there exists a branch $n_0(x_0)\srunp{*} n_1(x_1) \srunp{+} n_2(x_2)$ in $\FRT(\system,x_0)$ such that $x_1 < x_2$. Since $\system$ is unbounded, there are an infinite number of distinct states which are reachable from $x_0$. We first show that there exists a run starting from $x_0$, without any loop, i.e. where all states are distinct. We consider the finitely branching tree of all prefixes of runs, and prune this tree by removing prefixes which contain a loop. Because any reachable state can be reached without a loop, the pruned tree still contains an infinite number of prefixes. By K{\"o}nig's lemma, there exists an infinite run with no loop. Any run starting from $x_0$ has a finite prefix labelling a maximal path in $\FRT(\system,x_0)$. Hence, there must be a node $n_2(x_2)$ which is subsumed by a node $n_1(x_1)$ such that $x_1 \neq x_2$ (since we enforce that all nodes are distinct). Since we assumed $\leq$ to be a partial ordering, we deduce from 
			$x_1 \neq x_2$, and $x_1 \leq x_2$ that $x_1 < x_2$.
			
			Conversely, let us assume that there exist two vertices $n_1,n_2$ labelled by $x_1$ and $x_2$ respectively, in $\FRT(\system,x_0)$ such that 	$n_0(x_0)\srunp{*} n_1(x_1) \srunp{+} n_2(x_2)$ in $\FRT(\system,x_0)$ such that $x_1 < x_2$. Let $x_1 \srunp{a_1a_2...a_n} x_2$ and $x_1<x_2$. Hence, there exist $y_1, y_2,...,y_n \in \stateset$ such that $x_1 \srunp{a_1} y_1 \srunp{a_2} y_2 \srunp{a_3} ... \srunp{a_n} y_n =x_2$. 
			Since the system is strictly branch-monotone, we can repeat this sequence, \ie,  there exist $n$ states $u_1, u_2,...,u_n$ such that $x_2 \srunp{a_1} u_1  \srunp{a_2} u_2 ... \srunp{a_n} u_n = x_3$ with $y_n<u_n$. 
			By iterating this process, we construct an infinite sequence of states $(x_k)_{k \geq 0}$ such that for all $k \geq 1$, one has  $x_k \srunp{a_1a_2...a_n} x_{k+1}$ and $x_k < x_{k+1}$. Then, we deduce that all $x_k$ are different.
			Hence,  $\Post_{\system}^*(x_0)$ is infinite, and $\system$ is unbounded.	
		\end{proof}

		We now need a notion of effectivity adapted to branch-WSTS.
		
		\begin{defi}
			A branch-WSTS $\system= (X, \Sigma,$ $\srun, \leq, x_0)$ is \emph{branch-effective} if $\system$ is effective and $\Post_\system(x)$ is a (finite) computable set, for all $x \in X$. 
		\end{defi}
		
		\begin{thm}
			Boundedness is decidable for branch-effective branch-WSTS $\system = (X, \Sigma,$ $\srun, \leq, x_0)$ with strict branch-monotony such that $\leq$ is a partial ordering.
		\end{thm}
		
		\begin{proof} Suppose $\system = (X, \Sigma,$ $\srun, \leq, x_0)$ satisfies the above conditions.
			From Proposition~\ref{prop:finite}, we obtain that ${\FRT(\system,x_0)}$
			is finite. By hypothesis, $\system$ is finitely branching and branch-effective. In particular, for all $x$, $\Post_\system(x)$ is a finite computable set. As $\leq$ is decidable, 
			we deduce that
			$\FRT(\system,x_0)$
			is effectively computable.
			From Proposition~\ref{prop:boundedFRT}, we know that 
			$\Post^*_\system(x_0)$ is infinite \ifff there exists a finite branch $n_0(x_0)\srunp{*} n_1(x_1) \srunp{+} n_2(x_2)$ such that $x_1 < x_2$. This last property can be decided on $\FRT(\system,x_0)$, and so the boundedness property can be decided.
		\end{proof}
		
		We also generalise the decidability of non-termination for WSTS \cite{finkel_well-structured_2001} to branch-WSTS. 
		
		\begin{restatable}{prop}{termFRT}\label{prop:termFRT}
			A branch-WSTS $\system = (X, \Sigma,$ $\srun, \leq, x_0)$ does not terminate from state $x_0$ \ifff there exists a subsumed node in $\FRT(\system,x_0)$.
		\end{restatable}
		
		Hence, we have the following theorem:
		\begin{thm}
			Non-termination is decidable for branch-effective branch-WSTS.
		\end{thm}

		\begin{proof}
			Given a branch-WSTS $\system = (X, \Sigma,$ $\srun, \leq, x_0)$, we apply Proposition~\ref{prop:termFRT} so that it is sufficient to build $\FRT(\system,x_0)$ and check if there exists a subsumed node. Since $\system$ is branch-effective, 
			we can effectively construct $\FRT(\system,x_0)$ and verify the existence of a subsumed node.
		\end{proof}
		
		Note that the non-termination and boundedness problems for the example machine $\machine_1$ in Figure~\ref{fig:fifo-machine-1} are, therefore, decidable. Since there exist nodes $n_0(x_0)$ and $n_1(x_1)$ in the $\FRT$ such that $x_0 = (q_0, \epsilon)$ and $x_1 = (q_0, a)$ such that $x_0 <x_1$ and $x_0 \srunp{+} x_1$, the machine $\machine_1$ is unbounded. Furthermore, since all unbounded machines are non-terminating, 
		$\machine_1$ is non-terminating. 
		
		On the other hand, boundedness becomes undecidable if we relax the strict monotony condition to general monotony (even when we strengthen the order to be \wqo). This is because boundedness is undecidable for Reset Petri nets \cite{dufourd_reset_1998}. Reset Petri nets are effective WSTS $\system = (X, \Sigma,$ $\srun, \leq, x_0)$, hence branch-effective WSTS, where $\leq$ is the \wqo on vectors of integers. Hence, we deduce:
		
		\begin{prop}
			Boundedness is undecidable for branch-effective branch-WSTS $\system = (X, \Sigma, \srun, \leq, x_0)$ where $\leq$ is a \wqo.
		\end{prop}

		\section{Classes of branch-WSTS} \label{sec:examples}

		We now study two classes of transition systems, namely, counter machines and FIFO machines. Counter machines are \wqo for the extended natural ordering; however, due to the presence of zero tests, general monotony fails, and hence, they are not \WSTS. We look at a subclass of counter machines, called counter machines with restricted zero tests (\CMRZ), where zero tests are present, but their behaviour is controlled. 
		
		In \cite{brand_communicating_1983}, the authors showed that general FIFO machines can simulate Turing machines, and hence, it is undecidable to verify boundedness. However, it is still a widely used model to represent asynchronously communicating systems. Moreover, they are neither \wqos nor monotone for the prefix ordering. Secondly, in this section, we show that a subclass of FIFO machines, namely input-bounded FIFO machines, are branch-\WSTS under the prefix ordering.
		
		Finally, we study more general counter machines and FIFO machines, and show that we can decide non-termination for a large subclass using ideas from branch-well-structured behaviours.
		
		\subsection{Counter machines with restricted zero tests}  \label{sec:cmrz}
		
		We define \emph{counter machines with restricted zero tests (\CMRZ)} imposing the following requirement:
		Once a counter has been tested for zero, it cannot be incremented or decremented any more.
		Formally, we say that $\countermachine$ is a counter machine with restricted zero tests if for all transition sequences of the form $q_0 \srunp{op(\counter_1), Z_1} q_1 \srunp{op(\counter_2), Z_2} \ldots \srunp{op(\counter_{n-1}), Z_{n-1}} q_{n-1} \srunp{op(\counter_n), Z_n} q_n$, for every two positions $1 \leq i \leq j \leq n$, we have $\counter_j \not\in Z_i$.	
		
		In \cite{bollig_bounded_2020}, it was shown that non-termination and boundedness are decidable for this class of systems by reducing them to the (decidable) reachability problem. However, using the alternative approach of branch-WSTS, we can verify that  non-termination and boundedness are decidable for this class without reducing these problems to reachability.
		
		Given a \CMRZ $\countermachine =(Q,\counterset , T, q_0)$, we consider the associated transition system $\system_\countermachine = (X_\countermachine,\Action_\countermachine,\srun_\countermachine, x_0)$. From this system, we construct an \OLTS over the extended ordering $\leq$ such that $(q, \cscontents) \leq (q', \cscontents')$ \ifff $q = q'$ and $\cscontents \leq \cscontents'$ (component-wise). Note that $(X_\countermachine, \leq)$ is a \wqo. Moreover, this ordering is a partial ordering.
		
		We now show that \CMRZ are branch-monotone for $\leq$. We drop the subscript while denoting $X_\countermachine$ for the remainder of this section, as it is clear from context.
		
		\begin{prop}\label{prop:counterBranchComp}
			\CMRZ are branch-monotone and strictly branch-monotone for the extended ordering $\leq$.
		\end{prop}
		\begin{proof}
			
			Consider the \OLTS $\system = (X,\Action,\srun, \leq, x_0)$ associated to a \CMRZ with states $x, x' \in X$ such that $x \leq x'$ (resp. $x < x'$) and $x_0 \srunp{*} x \srunp{\sigma} x'$. We need to show that there exists a state $y$ such that $x' \srunp{\sigma} y$ and $x' \leq y$ (resp. $x'<y$). 
			
			We first prove the following claim:
			
			\noindent \emph{Claim: } For states $x, x' \in X$ such that $x_0 \srunp{*} x \srunp{\sigma} x'$ where $x \leq x'$ (resp. $x < x'$) and $|\sigma| = n$, the following property holds: For all $\nu \preceq \sigma$, we have $z, z' \in X$ such that $x \srunp{\nu} z$ and $x' \srunp{\nu} z'$ and $z \leq z'$ (resp. $z < z'$). We prove this claim by induction on the length of $\nu$. 
			
			For the base case, $|\nu| = 0$. We have from the hypothesis, $x \leq x'$, hence the claim is trivially true.
			
			Let us assume that the claim holds for $|\nu| = k$. We show that it holds for $|\nu| = k+1$. From the induction hypothesis, we know that for $\nu = \nu' \cdot a$ where $a \in \Action$, there exists $z_1, z'_1$ such that $x \srunp{\nu'} z_1$ and $x' \srunp{\nu'} z'_1$ and $z_1 \leq z'_1$. Since $x \srunp{\sigma} x'$, we know that there exists $z_2 \in X$ such that $x \srunp{\nu'} z_1 \srunp{a} z_2$. We can now be in one of the following cases:
			\begin{enumerate}[ label= \textbf{Case \roman*:}]
				 \item If $a$ is of the form $\noop$ and $Z = \emptyset$, then we can trivially execute $a$ from $z'_1$ and reach $z_2'$ such that $z_2 \leq z_2'$ (resp. $z_2 < z_2'$).
				\item The action $a$ is of the form $\inc{\counter}$ or $\dec{\counter}$, and the set of counters to be tested for zero $Z = \emptyset$, i.e. $z_1 \srunp{a} z_2$ only increments/decrements one counter and leaves the others unchanged (and no counters are tested to zero). Since $z_1 \leq z'_1$ (resp. $z_1< z'_1$), we know that $z_1, z'_1$ have the same control-state. Hence, this action is enabled in $z'_1$. Moreover, because of the \CMRZ property, we know that $\counter$ is not tested to zero even once until the state $z'_1$ is reached in this run. Therefore, we can execute the increment/decrement operation on $\counter$. Furthermore, since $z_1 \leq z'_1$ (resp. $z_1< z'_1$), the value of $\counter$ in $z'_1$ is greater than or equal to  (resp. strictly greater than) the value of $\counter$ in $z_1$. Hence, we can execute $a$ from $z'_1$ and reach a state $z_2'$ such that $z_2 \leq z_2'$ (resp. $z_2 < z_2'$).  
				\item $Z \neq \emptyset$ in the transition $z_1 \srunp{a} z_2$. Hence, there are a set of counters $Z$ which are tested to zero. By the \CMRZ property, we know that all counters $\counter \in Z$, are never incremented or decremented further. Hence, during the execution $z_1 \srunp{a} z_2 \srunp{w} z'_1$ where $w \in \Action^*$, we know that none of these counters are incremented or decremented. Hence, the value of the counters in $z'_1$ is also equal to zero. Therefore, we can execute $a$ from $z'_1$ to reach $z_2'$. Moreover, since $z_1 \leq z'_1$ (resp. $z_1 < z'_1$), and none of these counters change value, we can conclude that $z_2 \leq z_2'$ (resp. $z_2 < z_2'$).  
			\end{enumerate}
			
			Hence, as a special case of the claim where $\nu = \sigma$, we prove that \CMRZ are branch-monotone (resp. strictly branch-monotone).
		\end{proof}

		\begin{prop}\label{prop:counterBranchEff}
			\CMRZs are branch-effective branch-WSTS.
		\end{prop}
		\begin{proof}
			Given an \OLTS $\system = (X, A,$ $\srun, \leq, x_0)$ associated to a \CMRZ, for any two states $x, x' \in X$, we can decide if $x \leq x'$. Furthermore, $\srun$ is decidable from the underlying finite automaton, and $\Post_\system(x)$ is computable for all $x \in X$. Hence, it is branch-effective.
		\end{proof}
		
		Hence, we deduce:
		
		\begin{thm}\label{thm:cmrzbranch}
			Non-termination and boundedness are decidable for counter machines with restricted zero tests.
		\end{thm}

		\subsection{Input-bounded FIFO machines}\label{subsec:FIFO}
	
		The FIFO machine $\machine_1$ from Figure~\ref{fig:fifo-machine-1} is an example of a system that is branch-WSTS but the underlying set of states is not well-quasi-ordered. We first try to generalise a class of systems which are branch-\wqo, and which includes~$\machine_1$.

		We consider a restriction that has been studied in \cite{bollig_bounded_2020}, which we go on to prove is branch-\wqo. These systems are known as input-bounded FIFO machines, which we formally define below.

		Consider a FIFO machine $\machine = (Q, \alphabetset, T, q_0)$ over a set of channels $\chanset$. For $\ch \in \chanset$, we let $\sendproj{\ch}: \{\chanset \times \{!,?\} \times \alphabetset\}^* \to \alphabetset^*$ be the homomorphism defined by $\projsend{\ch! a}{\ch} = a$ for all $a \in \alphabetset$, and $\projsend{\alpha}{\ch} = \varepsilon$ if $\alpha$ is not
		of the form $\ch!a$ for some $a \in \alphabetset$. Moreover, $\projsend{\tau_1 \cdot \tau_2}{\ch} = \projsend {\tau_1}{\ch} \cdot \projsend{\tau_2}{\ch}$. Similarly, for $\ch \in \chanset$, we let $\recproj{\ch}: \{\chanset \times \{!,?\} \times \alphabetset\}^* \to \alphabetset^*$ be the homomorphism defined by $\projrec{\ch? a}{\ch} = a$ for all $a \in \alphabetset$, and $\projrec{\alpha}{\ch} = \varepsilon$ if $\alpha$ is not
		of the form $\ch?a$ for some $a \in \alphabetset$. Moreover, $\projrec{\tau_1 \cdot \tau_2}{\ch} = \projrec {\tau_1}{\ch} \cdot \projrec{\tau_2}{\ch}$.
		
		We define the input-language of a FIFO channel $\ch$ as the set of all words that are sent into the channel, i.e. $\projsend{\Traces{\machine}}{\ch}$. We say that the machine is \emph{input-bounded} if for each $\ch \in \chanset$, there is a regular bounded language $\IBL_\ch$ (i.e. language of the form $w_1^*\ldots w_n^*$) such that $\projsend{\Traces{\machine}}{\ch}\subseteq \Pref{\IBL_\ch}$, i.e. the send-projection of every run of the FIFO machine over each channel is a prefix of an input-bounded language. We say that $\IBL_\ch$ is distinct-letter if $|w_1 \ldots w_n|_a = 1$ for all $a \in \alphabetset$.
		
		\begin{prop} \label{prop:IBwqo}
			Input-bounded FIFO machines are branch-\wqo for the prefix-ordering $\leq_{p}$.
		\end{prop}
		\begin{proof}
			Let us consider the transition system $\system_{\machine} = (X_{\fifo},\Sigma_{\fifo},\srun_\fifo,x_0)$ associated an input-bounded FIFO machine $\fifo$ with a single channel $\ch$, and an infinite run $ x_0 \srun x_1\srun x_2\srun ... x_i...$ with $x_i=(q_i,w_i) \in X_{\fifo}$. 
			
			The infinite run is of the form $x_0 \srunp{\sigma_1} x_1 \srunp{\sigma_2} x_{2}... x_{i-1} \srunp{\sigma_i} x_i \srunp{\sigma_{i+1}} \ldots $ and we denote
			$\sigma[i] = \projsend{\sigma_1\sigma_2\ldots \sigma_i}{\ch}$.
			It can be observed that $\sigma[i]$ is a prefix of $\sigma[i+1]$ for all $i \in \mathbb{N}$. Since $\sigma[i]$ is of the form $v_1^{n_{1,i}}...v_m^{n_{m,i}}u_m$ 
			for $u_m \preceq v_m$ and $n_{1,i}..., n_{m,i} \geq 0$ and $1 \leq m \leq k$, the infinite sequence $(\sigma[i])_{i \in \mathbb{N}}$ satisfies two possible exclusive cases:
			\begin{enumerate}[ label= \textbf{Case \roman*:}]
				
				\item There exists an $i_0$ such that $\forall i \geq i_0, \projsend{\sigma_i}{\ch}=\epsilon$ so there exists $i_1 \geq i_0$ such that for all $i \geq i_1, w_i=w_{i+1}$. Hence, because there are finitely many control-states, we deduce that there exist $i_2,i_3\geq i_1$ such that  $x_{i_2} \srunp{+} x_{i_3}$ and $x_{i_2}=x_{i_3}$ hence also in particular $x_{i_2}  \leq_{p} x_{i_3}$.
				
				\item There are infinitely many indices $i$ such that $\projsend{\sigma_i}{\ch} \neq \epsilon$, which means that the infinite sequence $(\sigma[i])_{i \in \mathbb{N}}$ is not stationary. This implies that the set $S_{\sigma}=\{(n_{1,i}, ... , n_{k,i})  \mid i \in \mathbb{N}\}$, associated with $\sigma$, is infinite. Hence, there exists a least index $p$ such that the set $\{n_{p,i}\}_{i \in \mathbb{N}}$ is infinite. Then the set $F=\{(n_{1,i}, ... , n_{p-1,i})  \mid i \in \mathbb{N} \}$ is finite.

				We claim that for all indices $\ell \geq p+1$, $n_{\ell,i}  = 0$ for all $i$. Let us assume to the contrary that there is some index $\ell \geq p+1$ and $i_0$ such that $n_{\ell, i_0} \neq 0$. This means that the word $v_{\ell}$ is in the channel in state $x_{i_0}$, which means that the word $v_{\ell}$ was sent to the channel before (or at) the step $i_0$, i.e, $\sigma[i_0] = v_1^{n_{1,{i_0}}}...v_p^{n_{p,{i_0}}}...v_m^{n_{m,{i_0}}}u_m$ for some $u_m \preceq v_m$ and $n_{\ell,{i_0}} > 0$ and $1 \leq m \leq k$. So, in particular, word $v_p$ cannot be sent after $i_0$, hence, $n_{p,i} = n_{p,{i_0}}$ $ \forall i>i_0.$ Hence,  $\{n_{p,i}\}_{i \in \mathbb{N}}$ is finite which is a contradiction to our assumption that  $\{n_{p,i}\}_{i \in \mathbb{N}}$ is infinite. 
				
				This means that after some state $x_t$, we only write word $v_p$ to the channel. Since, the set $F=\{(n_{1,j}, ... , n_{p-1,j})  \mid j \in \mathbb{N} \}$ is finite, we can extract an infinite subsequence $ (q,w_i)_{i \in K \subseteq \mathbb{N}}$ where
				$w_i = uv_p^{n_{p,i}}$ with $u \in F$ and $(n_{p,i})_{i \in K}$ is non-decreasing. Hence, there exist two indices $a,b > 0$ such that $w_a = u.v_p^{n_{p,a}}$ and $w_{a+b} = u.v_p^{n_{p,a+b}}$ and $n_{p,a} \leq n_{p,a+b}$ hence $w_{a+b}=w_a.v_p^{n_{p,a+b}-n_{p,a}}$  hence $w_{a} \preceq w_{a+b}$. So we have found two states $x_a, x_{a+b}$ such that $x_a \leq_{p} x_{a+b}$. Hence, the machine is branch-\wqo for the prefix ordering.
				
			\end{enumerate}
			
			Using the same argument for each channel, we can conclude that input-bounded FIFO machines with multiple channels are branch-\wqo for the extended prefix-ordering $\leq_{p}$.
		\end{proof}
		It is clear that $\machine_1$ from Figure~\ref{fig:fifo-machine-1} belongs to this class of FIFO machines. But, we see below that the class of input-bounded FIFO machines is not branch-WSTS.
		\begin{exa}\label{ex:input-fifo-1} Consider the FIFO machine $\machine_3$ with a single channel $\ch$ in Figure~\ref{fig:input-fifo-1} that is (distinct-letter) input-bounded for $\IBL_\ch = (ab)^*$. We have 
			$(q_0, \epsilon) \srunp{\sigma} (q_0,b)$, where $\sigma = !a!b?a$. Moreover, $(q_0, \epsilon) \extpref (q_0, b)$. However, we cannot repeat $\sigma$ from $(q_0, b)$, as it is not possible to execute $(q_2, bab) \srunp{?a}$. Hence, the machine is not branch-monotone for the prefix-ordering.
		\end{exa}

		\begin{figure}[!ht]
			\centering
				\begin{tikzpicture}[->, auto, thick]

					\node[place, initial, initial text=] (p2) [right= of p1]{$q_0$};
					\node[place] (p3) [right= of p2] {$q_1$};
					\node[place] (p4) [right= of p3] {$q_2$};

					\path[->]
					(p2) edge node[] {$!a$} (p3)
					
					(p4) edge [bend right=40] node[swap] {$?a$} (p2)
					(p3) edge node[] {$!b$} (p4)
					;
					
				\end{tikzpicture}
			\caption{The FIFO machine $\machine_3$\label{fig:input-fifo-1}}
			
		\end{figure}
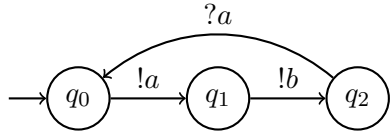

		From the above counter-example, we see that the class of distinct-letter input-bounded FIFO machines are not branch-WSTS. Hence, we impose another restriction on such systems. 
		
		Consider an input-bounded FIFO machine $\hat \machine = (\hat Q, \alphabetset, \hat T, q_0)$ (over a set of channels $\chanset$) with a distinct-letter bounded input-language $\IBL = (\IBL_\ch)_{\ch \in \chanset}$. We first consider a deterministic
		 finite automaton $\A = (Q_\A,\fifoAlpha{\fifo},\fifotransitions_\A,q^0_\A,F_\A)$,
		with set of final states $F_\A \subseteq Q_\A$, whose language is $\EL(\A) =  \sendL \mathrel{\cap} \Pref{\recL}$, where $\sendL = \{\sigma \mid \projsend{\sigma}{\ch} \in \IBL_\ch \text{ for all } \ch \in \chanset\}$ and $\recL = \{\sigma \mid \projrec{\sigma}{\ch} \in \IBL_\ch\text{ for all } \ch \in \chanset\}$. 
		
		With this, we define $ \bar\fifo_{\IBL} =  (Q, A, T, q_0)$ as the product of the FIFO machine $\hat \fifo$ and $\A$ in the expected manner.
		In particular, the set of control-states of $\bar\fifo_{\IBL}$ is $Q =  \hat Q \times Q_\A$, and its initial state
		is the pair $q_0 = (\hat{q_0},q^0_\A)$.
		We assume we are given an input-bounded FIFO machine $\hat\fifo$ and its input language $\IBL$. Let $\bar\fifo_{\IBL}$ be the FIFO machine constructed as above. Then:
		
		\begin{prop}
			The machine $\bar\fifo_{\IBL}$ is branch-monotone.
		\end{prop}
		
		\begin{proof}
			We show this proof for a FIFO machine with a single channel $\ch$, but the same argument can be extended to FIFO machines with multiple channels.
			
			Let $\mathcal{L} = w_1^*\ldots w_n^*$. Let $(q_0, \varepsilon) \srunp{\tau} (q, w) \srunp{\sigma} (q,w')$ such that $w \preceq w'$. To prove branch-monotony, we need to show that there exists $w''$ such that $(q, w') \srunp{\sigma} (q, w'')$ and $w' \preceq w''$.
			
		 Since $\IBL$ is a bounded language, we know that $\projsend{\tau}{\ch} = w_1^{n_1}\ldots w_i^{n_i}. u_i$ where $u_i \preceq w_i$ and $1 \leq i \leq n$ and $n_p\in \mathbb{N}$ for all $1 \leq p \leq i$. Moreover, $\projsend{\sigma}{\ch} \in \Pref{u'_i \cdot w_i^{n'_i} \ldots w_j^{n_j}}$ where $u_i. u'_i = w_i$ and $1 \leq i \leq j \leq n$. Let us consider the channel content $w$ now. From the characterisation of $\tau$ above, we can express $w = v_\ell \cdot w_{\ell}^{n_\ell} \ldots w_i^{n_i} . u_i$, where $1 \leq \ell \leq i$ and $v_\ell \in \Suf{w_\ell}$. Now, let us analyse the cases based on the value of $\projrec{\sigma}{\ch}$: 
			\begin{enumerate}[ label= \textbf{Case \roman*:}]
				\item $\projrec{\sigma}{\ch} = \varepsilon$. In other words, this means that there are only send actions in $\sigma$. Hence, it is possible to repeat the same sequence of moves as we are in the same control-state $q$. Therefore, $(q, w') \srunp{\sigma} (q,w'')$ for some value of $w''$.  Furthermore, since $\sigma$ has only send actions, $w' = w.v$ for some $v \in A^*$. Therefore, after we repeat $\sigma$ once again from $(q,w')$, we reach $(q, w'')$ such that $w'' = w' . v = w. v. v$. Therefore, $(q, w') \leq_{p} (q, w'')$ and we are done with this case. 
				
				\item $\projrec{\sigma}{\ch} \neq \varepsilon$.
				\begin{enumerate}[ label= \textbf{Case \alph*:}]
				 \item Let us first consider that  $w \neq \varepsilon$. 
				 
				 Let us first assume $\exists p_1, p_2 \cdot~ 1 \leq p_1 < p_2 \leq i$ such that $w = v_{p_1} . v . u_{p_2}$ where $v_{p_1} \in \Suf{w_{p_{1}}}$, $u_{p_2} \in \Pref{w_{p_2}}$, and $v_{p_1}, u_{p_2} \neq \varepsilon$ and $v \in \alphabetset^*$. In other words, there is a non-empty part of at least two distinct words in the channel contents $w$. Since the FIFO machine is input-bounded, we can conclude that $\projsend{\sigma}{\ch}$ does not contain any occurrences of alphabets from the word $w_{p_1}$. Therefore, in order for the condition $w \preceq w'$ to be satisfied, it is necessary that $\projrec{\sigma}{\ch} = \varepsilon$, which is a contradiction to our assumption. 
				
					 Therefore, if $w \neq \varepsilon$, then the only possibility is that $w = v_i. w_i^{n_i}. u_i$ such that $v_i \in \Suf{w_i}$. Therefore, $\projrec{\sigma}{\ch}$ only consists of letters from words $w_j$ such that $j \geq i$. However, since $w \preceq w'$, we can be certain that it only consists of letters from the word $w_i$ (if the head of the channel contains letters from a word $w_j$ where $j>i$ from the channel, then there can be no occurrence of the word $w_i$ in the channel). Therefore, $\projrec{\sigma}{\ch}$ consists of only letters belonging to $w_i$. Moreover, since $\projsend{\sigma}{\ch}$ is non-empty, there is at least one letter that is read from $w$. Therefore, the first letter that is sent in the sequence $\sigma$ belongs to the word $w_i$ (to ensure $w \preceq w'$). 
					
					Let us consider this subsequence $\sigma'$ from $(q,w)$ to the first send action. Let us say we have $(q,w) \srunp{\sigma'} (q',v')$. Now, since the subsequence $\sigma'$ only consists of receptions from $(q,w)$, along with the first send action, this subsequence is also possible from $(q, w.v)$ for all $v \in \Sigma^*$. Therefore, we can execute the same sequence from $(q, w')$. Hence, $\projsend{\tau . \sigma . \sigma'}{\ch} \in \IBL$. Therefore, since $\mathit{Alph}(\projsend{\sigma'}{\ch}) \in \mathit{Alph}{(w_i)}$, we can be sure that $\mathit{Alph}{(\projsend{\sigma}{\ch})} \in \mathit{Alph}{(w_i)}$. Therefore, $\sigma$ only sends and receives letters from a single word $w_i$. 
					
					Moreover, since the system is input-bounded, and the first send action in $\sigma'$ matches the first send action in $\sigma$, we see that $w' = v_i . w_i^{n_i}. u_i . (v'_i . w_i^{n'_i}. u_i)$ ${ = w. (v'_i . w_i^{n'_i}. u_i)}$ such that $u_i.v'_i = w_i$. Therefore, we can repeat this sequence from $(q, w')$ and reach a state $(q, w'')$ such that $w' \preceq w''$, and hence, it is branch-monotone for this case.
					
					\item The final case we need to consider is $w = \varepsilon$. In this case, it is clear that $\sigma$ consists of at least one send action before the first reception. Therefore, because of the input-bounded property and the fact that this action can be executed at $(q, w')$, we can once again see that $\projsend{\sigma}{\ch}$ consists only sending only letters from a single word. Moreover, since the same action can be executed, once again we see that $\projsend{\sigma}{\ch} = v_j . w_j^{n_j} . u_j$ such that $u_j. v_j = w_j$. Therefore, $\projrec{\sigma}{\ch} \in \Pref{v_j . w_j^{n_j} . u_j}$. 
					
					Now let us consider the run $\tau.\sigma$ in the automaton $\A$ that we constructed. Since $\tau.\sigma$ is a run in $\bar\fifo_{\IBL}$, there is also a run in $\A$ such  that $q^0_\A \srunp{\tau} q_s \srunp{\sigma} q_s$. Moreover, we can also repeat $\sigma$ to obtain $q^0_\A \srunp{\tau} q_s \srunp{\sigma} q_s \srunp{\sigma} q_s$. Therefore, $\tau.\sigma.\sigma \in \Pref{\recL}$. Moreover, since $\projrec{\sigma}{\ch} \neq \varepsilon$, $\projrec{\sigma}{\ch} = u'_j .w_j^{n'_j}.v'_j$ such that $v'_j.u'_j = w_j$. Therefore, we can repeat $\sigma$ from $(q,w')$ in $\bar\fifo_{\IBL}$, and we reach a state $(q,w'')$ such that $w' \preceq w''$.

				\end{enumerate}
				
			\end{enumerate}
			
			Hence, we see that for all cases, if $(q_0, \varepsilon) \srunp{\tau} (q, w) \srunp{\sigma} (q,w')$ such that $w \preceq w'$, then there exists $w''$ such that $(q, w') \srunp{\sigma} (q, w'')$ and $w' \preceq w''$. 
			
			If we have more than one channel, we can extend this argument to consider each channel, and hence, $\bar\fifo_{\IBL}$ is branch-monotone.
		\end{proof}
		\begin{exa} Consider the FIFO machine $\fifo_4$ as in Figure~\ref{fig:input-fifo-2} that is input-bounded for $\IBL = (ab)^*$. As we saw in Example~\ref{ex:input-fifo-1}, it is not branch-monotone. However, let us consider the product $\bar\fifo_{4,\IBL}$ of $\fifo_4$ along with the finite automaton $\A$ that recognises $\sendL \cap \Pref{\recL}$. We depict only
			accessible states of $\fifo_{4,\IBL}$, from which we can still complete the word read so far to a word belonging to $\sendL \cap \Pref{\recL}$. Here, we see that the run we had in counter-example previously no longer violates branch-monotony. In fact, the loop that could not be realised has now been unfolded, and we can see that the FIFO machine $\fifo_4$ has only a finite run. Trivially, due to the absence of loops, we see that it is now branch-monotone for the prefix-ordering.
		\end{exa}
		
		\newcommand{\sname}[3]{\begin{array}{c}#1\\[-0.8ex]{\scalebox{0.7}{#2,#3}}\end{array}}

		\begin{figure}[!ht]
			\centering
			\scalebox{0.9}{
				\begin{tikzpicture}[->, auto, thick]
						\begin{scope}
						\node[place, initial, initial text=] (p2){$q_0$};
						\node[place] (p3) [right= of p2] {$q_1$};
						\node[place] (p4) [right= of p3] {$q_2$};

						\path[->]
						(p2) edge node[] {$!a$} (p3)
						
						(p4) edge [bend right=40] node[swap] {$?a$} (p2)
						(p3) edge node[] {$!b$} (p4)
						;
						\node[][below= 0.6cm of p3]{ $\fifo_4$ };
					\end{scope}
					\begin{scope}[shift={(8,2)}]
						\node[place, initial, initial text=,accepting] (q0) {$s_{0}$};
						\node[place] (q1) [right= of q0] {$s_1$};
						\node[place,initial, initial text=,accepting] (q2) [below= 2.5cm of q0] {$r_0$};
						\node[place,accepting] (q3) [right= of q2] {$r_1$};
						\path[->]
						(q0) edge[bend left] node[] {$!a$} (q1)
						(q1) edge[bend left] node[] (l2){$!b$} (q0)
						(q2) edge[bend left] node[] {$?a$} (q3)
						(q3) edge[bend left] node[] (l1) {$?b$} (q2)
						;
						\node[][below= 0.2cm of l2]{ $L(\mathcal{A}_s)=\sendL$ };
						\node[][below= 0.2cm of l1]{ $L(\mathcal{A}_r)=\Pref{\recL}$ };
					\end{scope}

					\begin{scope}[shift={(0,-5)}]
						\node[place, initial, initial text=, text width=0.7cm, minimum size=1.4cm, align=center] (r0){$q_0$, $s_0, r_0$};
						\node[place, text width=0.7cm, minimum size=1.4cm, align=center] (r1) [right= of r0, minimum size=1.4cm] {$q_1$, $s_1, r_0$};
						\node[place, text width=0.7cm, minimum size=1.4cm, align=center] (r2) [right= of r1, minimum size=1.2cm] {$q_2$, $s_0, r_0$};
						\node[place, text width=0.7cm, minimum size=1.4cm, align=center] (r3) [right= of r2,minimum size=1.2cm] {$q_0$, $s_0, r_1$};
						\node[place, text width=0.7cm, minimum size=1.4cm, align=center] (r4) [right= of r3,  minimum size=1.2cm] {$q_1$, $s_1, r_1$};
						\node[place, text width=0.7cm, minimum size=1.4cm, align=center] (r5) [right= of r4,  minimum size=1.2cm] {$q_1$, $s_1, r_1$};
						\path[->]
						(r0) edge node[] {$!a$} (r1)
						
						(r1) edge node[] {$!b$} (r2)
						(r2) edge node[] {$?a$} (r3)
						(r3) edge node[] {$!a$} (r4)
						(r4) edge node[] {$!b$} (r5)
						;
						\node[][below= 0.5cm of r2]{ $\bar\fifo_{4,\IBL}$ };
					\end{scope}
					
				\end{tikzpicture}}
			\caption{The FIFO machine $\machine_4$ (right) and the automata $\mathcal{A}_s$ and $\mathcal{A}_r$ (left) that recognize $\sendL$ and $\Pref{\recL}$ respectively intersect to give $\bar\fifo_{4,\IBL}$(below) which is branch-monotone. \label{fig:input-fifo-2}}
			
		\end{figure}
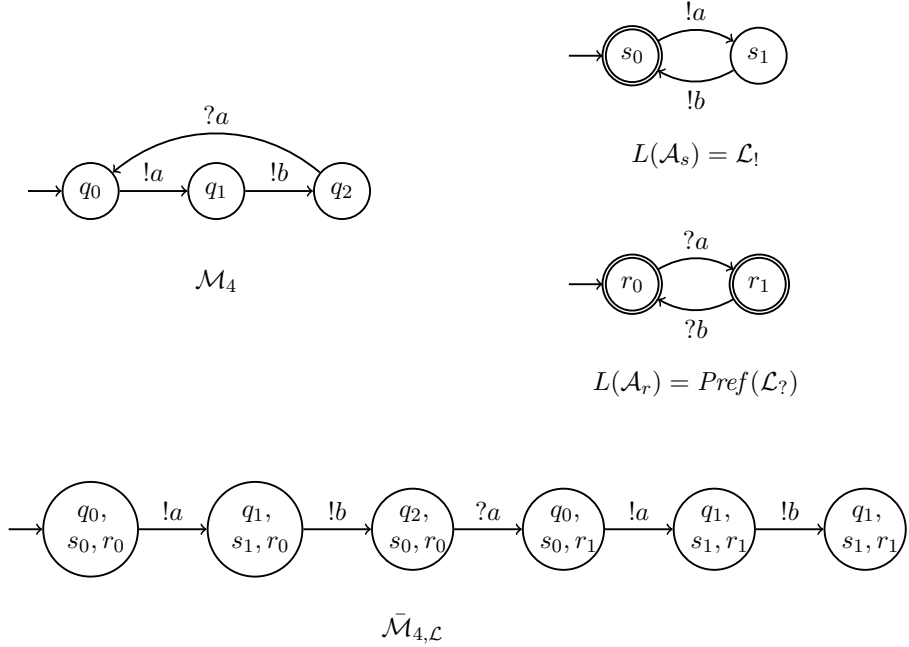

		We remark here that in \cite{bollig_bounded_2020}, when the authors show that input-bounded rational-reachability is decidable, they construct a ``normalised" FIFO machine, from the given FIFO machine and bounded language. Using the same principles, we can modify every input-bounded FIFO machine into one that is distinct-letter, and using the product construction from above, we have the following proposition:
		
		\begin{prop}\label{prop:ibFifoBranchEff}
			Normalised input-bounded FIFO machines are branch-effective branch-WSTS (for the prefix ordering).
		\end{prop}
		\begin{proof}
			Given an input-bounded FIFO machine $\system_\machine$, for any two states $x, x' \in X$, we can decide if $x \extpref x'$. Furthermore, $\srun$ is decidable, and $\Post_\system(x)$ is computable for all $x \in X$. Hence, it is branch-effective. 
		\end{proof}
		
		Moreover, the extended prefix ordering is a partial ordering. Hence, we deduce:
		\begin{thm}\label{thm:fifobranch}
			Non-termination and boundedness are decidable for input-bounded FIFO machines.
		\end{thm}

	\subsection{Verifying non-termination for general counter machines}
We saw in Section~\ref{sec:cmrz} that non-termination and boundedness are decidable for \CMRZ. However, this class of counter machines is weak as every run in such a machine has to maintain the restriction. In this section, we show that we can study non-termination for a larger class of counter machines, where the branch-well-structured behaviour is only maintained for a single run.

We begin by constructing the labelled-$\FRT$ (denoted by $\lRRT$), which is a refinement of the $\FRT$, where we add an additional label, which we will call \emph{iterable}. We mark a dead vertex $n'$ (labelled by $x'$) as \emph{iterable} if \begin{itemize}
	\item it is subsumed by a node $n$ labelled by $x$, i.e. $x \leq x'$ and $\exists \sigma. x \srunp{\sigma} x'$, and
	\item there exists a state $x'' \in X$, such that, $x' \srunp{\sigma} x''$ and $x' \leq x''$.
	
\end{itemize}

For counter machines, we show that the presence of an iterable node in the $\lRRT$ is a sufficient condition for non-termination.

\begin{prop}
	A counter machine  	$\system = (X,\Action,\srun, \leq, x_0)$ is non-terminating if \linebreak $\lRRT(\system, x_0)$ has at least one iterable node.
\end{prop}
\begin{proof}
	Let $\system = (X,\Action,\srun, \leq, x_0)$  be the \OLTS associated  to a counter machine, and $\lRRT(\system, x_0)$ its labelled reduced reachability tree. We show that if there exists an iterable node in $\lRRT(\system, x_0)$, then $\system$ is non-terminating. Let us assume there is an iterable node labelled by $x_1$ in $\lRRT(\system, x_0)$. By definition, this means that there exists a run $x_0 \srunp{*} x_1\srunp{\sigma} x_2 \srunp{\sigma} x_3$. 
	
	Let $x_i = (q, \cscon_{i})$ for $i \in \{1, 2, 3\}$, where $\counterval_i = (\counterval_{i,\counter})_{\counter \in \counterset}$
	represents the counter values. 
	Let $\sigma = a_1 a_2 \ldots a_{n}$. Then we have, $x_1 \srunp{a_1} u_1 \srunp{a_2} \ldots \srunp{a_{n-1}} u_{n-1} \srunp{a_{n}} x_2$ and $x_2 \srunp{a_1} u_1' \srunp{a_2} \ldots \srunp{a_{n-1}} u_{n-1}' \srunp{a_{n}} x_3$. Let $u_i = (q_i, \csconp_{i})$, and $u_i' = (q_i, \csconp'_{i})$ for $1 \leq i < n$,  where $\csconp_{i} = (\csconp_{i,\counter})_{\counter \in \counterset}$ and $\csconp'_{i} = (\csconp'_{i,\counter})_{\counter \in \counterset}$
	represent the counter values. 
		
	 Since $x_1 \leq x_2$, for each counter $\counter \in \counterset$, the following holds: $\counterval_{2,\counter} = \counterval_{1,\counter} + k_\counter$, where $k \geq 0$. Moreover, because both runs perform the same sequence of actions, we also have the following property: for all $1 \leq i \leq n$,  $\csconp'_{i,\counter} = \csconp_{i,\counter} + k_\counter$. 
	 
	 Let us assume that there exists a transition $u_{i-1} \srunp{a_{i}} u_{i}$, such that $a_i = \op(\counter_i), Z_i$, 
	 for some $\counter_i \in \counterset$, and where $Z_i \neq \emptyset$, then we also have $u'_{i-1} \srunp{a_{i}} u'_{i}$. Therefore, the set of counters $Z_i$ to be tested to zero in $u_{i-1}$ is the same as the set of counters to be tested to zero in $u'_{i-1}$. And the action $a_i$ is feasible from both $u_i$ and $u_i'$. Therefore, for each counter $\counter \in Z_i$, we have $\csconp_{i,\counter} = \csconp'_{i,\counter}$, i.e. $k_\counter = 0$. In other words, for each counter $\counter \in \counterset$ that is tested to zero along $\sigma$, $k_\counter = 0$, i.e. the values of those counters are identical along both runs. 
	 
	 Let us now only consider the set of counters tested for zero along $\sigma$, and call this set of counters $Z_\sigma$. The following property holds: $\counterval_{1,\counter} = \counterval_{2,\counter}$ for all $\counter \in Z_\sigma$. Now, from $(q, \counterval_2)$, upon executing the sequence of transitions $\sigma$, for every counter in $Z_\sigma$, once again the values do not change. Therefore, we have the following property: $\counterval_{2,\counter} = \counterval_{3,\counter}$ for all $\counter \in Z_\sigma$. For counters
	 $\counter \in \counterset \setminus Z_\sigma$, $\counterval_{2,\counter} = \counterval_{1,\counter} + k_\counter$. However, since there are no zero tests for these counters along $\sigma$, we can repeat the same sequence of actions and obtain: $\counterval_{3,\counter} = \counterval_{2,\counter} + k_\counter$.
	 
	 Therefore, from $x_3 = (q, \cscon_3)$, once again we can repeat $\sigma$ in order to reach a state $x_4 = (q, \cscon_4)$ such that for all $\counter \in Z_\sigma$, $\counterval_{3,\counter} = \counterval_{4,\counter}$ and for all counters $m \in \counterset \setminus Z_\sigma$, $\counterval_{4,m} = \counterval_{3,m} + k_m$. Therefore, we can repeat $\sigma$ infinitely many times. Hence, $\system$ is non-terminating.
\end{proof}
		
		\subsection{Verifying non-termination for FIFO machines}
		
		Now, we extend the same idea as above for branch-\wqo FIFO machines.

		Let us assume that given a FIFO machine, we have: $(q,u) \srunp{\sigma} (q, u.v)$, and  $\exists w \in \labelset^*$ such that $(q,u.v) \srunp{\sigma} (q,u.v.w)$. Then, for all $\ch \in \chanset$:
		
		\begin{align}
			u_\ch. \projsend{\sigma}{\ch} = \projrec{\sigma}{\ch}.u_\ch.v_\ch\\
			u_\ch.v_\ch.\projsend{\sigma}{\ch} = \projrec{\sigma}{\ch}.u_\ch.v_\ch.w_\ch
		\end{align} 
		From the above two equations, we have
		\begin{equation}
			u_\ch.v_\ch.\projsend{\sigma}{\ch} = u_\ch.\projsend{\sigma}{\ch}.w_\ch
		\end{equation}
		Hence, we have:
		\begin{equation}
			v_\ch.\projsend{\sigma}{\ch} = \projsend{\sigma}{\ch}.w_\ch
		\end{equation}
		Moreover,
		\begin{align}
			|\projsend{\sigma}{\ch}| = |\projrec{\sigma}{\ch}| + |v_\ch|\\
			|\projsend{\sigma}{\ch}| = |\projrec{\sigma}{\ch}| + |w_\ch|
		\end{align} 
		Hence, the length of $\projsend{\sigma}{\ch}$ is at least as much as that of $v_\ch$ and $w_\ch$. Moreover, 
		\begin{equation}
			|v_\ch| = |w_\ch|
		\end{equation}
		Moreover, from Equation 4.1, we see that $v_\ch$ is a suffix of $\projsend{\sigma}{\ch}$. Similarly, from Equation 4.2, we see that $w_\ch$ is a suffix of $\projsend{\sigma}{\ch}$. Since $|v_\ch| = |w_\ch|$, we have $v_\ch = w_\ch$. We can now rewrite Equation 4.4 as:
		\begin{equation}
			v_\ch.\projsend{\sigma}{\ch} = \projsend{\sigma}{\ch}.v_\ch
		\end{equation}
		
		The well-known Levi’s Lemma says that the words $u, v \in \labelset^*$
		that are solutions of the equation $uv = vu$ satisfy $u, v \in z^*$ where $z \in \labelset^*$ is a primitive word.

	 For FIFO machines also, we show that an iterable node in the $\lRRT$ can indicate the presence of a non-terminating run. Before we prove that, we first restate a result from \cite{finkel_verification_2020} that we will use.

		\begin{propC}[\cite{finkel_verification_2020}]\label{prop:fp19}
			Given a FIFO machine, the loop labelled by $\sigma$ is infinitely iterable from a state $(q, \chcontents)$ \ifff for every channel $\ch \in \chanset$, either $\projrec{\sigma}{\ch} = \varepsilon$, or the following three conditions are true: $\sigma$ is fireable at least once from $(q, \chcontents)$, $|\projrec{\sigma}{\ch}| \leq |\projsend{\sigma}{\ch}|$ and $\chcontents_\ch\cdot(\projsend{\sigma}{\ch})^\omega = (\projrec{\sigma}{\ch})^\omega$.
			\end{propC}
			
			Now, we will use the above result to provide a \emph{sufficient} condition for non-termination in FIFO machines.

\begin{prop}
	\sloppy A FIFO machine  	$\system = (X,\labelset,\srun, \leq_p, x_0)$ is non-terminating if ${\lRRT}{(\system, x_0)}$ has at least one iterable node.
\end{prop}	
\begin{proof}
	Let $\system = (X,\labelset,\srun, \leq_p x_0)$  be a FIFO machine, and $\lRRT(\system, x_0)$ its labelled reachability tree.	Let us assume that there exists an iterable node $n'$ labelled by $x'$ in $\lRRT(\system, x_0)$. Then, by definition, we have a run $x_0 \srunp{*} x \srunp{\sigma} x' \srunp{\sigma} x''$ such that $x \leq_p x' \leq_p x''$. Let $x = (q, \chcontents)$, $x' = (q, \chcontents')$ and $x'' = (q, \chcontents'')$. From Equations 4.1 through 4.7, we know that for all $\ch \in \chanset$, if $\chcontents_\ch = u_\ch$, then $\chcontents'_\ch = u_\ch \cdot v_\ch$, and $\chcontents''_\ch = u_\ch \cdot v_\ch \cdot v_\ch$, for some words $u_\ch, v_\ch \in \labelset^*$.
	
	In order to show that $\system$ is non-terminating, we will try to demonstrate the criteria in Proposition~\ref{prop:fp19}. We will show that $\sigma$ is infinitely iterable from $x = (q, \chcontents)$. For each $\ch \in \chanset$ such that $\projrec{\sigma}{\ch} = \varepsilon$, the condition in the proposition is true trivially. Hence, we only need to show that in cases when $\projrec{\sigma}{\ch} \neq \varepsilon$, the three conditions stated are met. We will consider a single channel, but the idea can be extended to all channels.
	
	Since we have a run $x_0 \srunp{*} x \srunp{\sigma} x'$, the sequence of actions $\sigma$ is fireable at least once, so the first condition holds. Moreover, since $(q, \chcontents) \srunp{\sigma} (q, \chcontents')$ and for every $\ch \in \chanset$, $\chcontents_\ch \pref \chcontents'_\ch$, the net effect of the loop is not decreasing in the length of the channel contents, i.e. $|\projrec{\sigma}{\ch}| \leq |\projsend{\sigma}{\ch}|$. Hence, the second condition also holds true.
	
	We only now need to show that for all $\ch \in \chanset$, the following holds: $u_\ch\cdot(\projsend{\sigma}{\ch})^\omega = (\projrec{\sigma}{\ch})^\omega$. From Equation~4.1, we have: \begin{equation*}
		u_\ch. \projsend{\sigma}{\ch} = \projrec{\sigma}{\ch}.u_\ch.v_\ch
	\end{equation*}
	Concatenating both sides with the infinite word $\projsend{\sigma}{\ch})^\omega$, we have: \begin{align}
u_\ch. (\projsend{\sigma}{\ch})^\omega &= \projrec{\sigma}{\ch}.u_\ch.v_\ch. (\projsend{\sigma}{\ch})^\omega \\
 &= \projrec{\sigma}{\ch}.u_\ch.v_\ch. \projsend{\sigma}{\ch}.(\projsend{\sigma}{\ch})^\omega
	\end{align}
	Moreover, from Equation~4.2, we now have:
	\begin{align}
		u_\ch. (\projsend{\sigma}{\ch})^\omega &= \projrec{\sigma}{\ch}.\projrec{\sigma}{\ch}.u_\ch.v_\ch.w_\ch. (\projsend{\sigma}{\ch})^\omega \\
		&= (\projrec{\sigma}{\ch})^2.u_\ch.v_\ch.v_\ch. (\projsend{\sigma}{\ch})^\omega
	\end{align} 
	Which we can rewrite using Equation~4.8 as:
	\begin{align}
		u_\ch. (\projsend{\sigma}{\ch})^\omega &= (\projrec{\sigma}{\ch})^2.u_\ch.v_\ch.\projsend{\sigma}{\ch}.v_\ch.(\projsend{\sigma}{\ch})^\omega
	\end{align}
	And now, we can repeat the process, and use Equation~4.2 once again to obtain:
	\begin{align}
u_\ch. (\projsend{\sigma}{\ch})^\omega &= (\projrec{\sigma}{\ch})^2.\projrec{\sigma}{\ch}.u_\ch.v_\ch.w_\ch.v_\ch.(\projsend{\sigma}{\ch})^\omega
	\end{align}
	Repeating this process, we have: $u_\ch\cdot(\projsend{\sigma}{\ch})^\omega = (\projrec{\sigma}{\ch})^\omega$. Hence, the presence of an iterable node in the $\lRRT$ implies infinite iterability, or in other words, non-termination.
	\end{proof}
	
		As for the case of counter machines, the $\lRRT$ of branch-\wqo FIFO machines is finite. Hence, for branch-\wqo FIFO machines, it is decidable to verify if there is an iterable node in the $\lRRT$. 
		
		\section{Decidability of Coverability} \label{sec: coverability}
		
		\noindent\textbf{Coverability algorithms for branch-WSTS.}
		We show that the two existing coverability algorithms (the forward and backward algorithms) for WSTS (see Section~\ref{subsec:decprobs} to recall) do not allow one to decide coverability for branch-WSTS. We remark that, contrary to WSTS, ${\Pre^*(\upclose x)}$ is not necessarily upward-closed for branch-WSTS. In fact, even with a single zero test, this property may not satisfied. 
		
		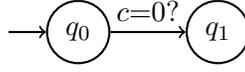
\begin{figure}[t]
			\centering
			\begin{tikzpicture}[->, auto, thick]
				\node[place, initial, initial text=] (p1) {$q_0$};
				\node[place] (p2) [right= of p1]{$q_1$};

				\path[->]
				(p1) edge node[] {$\counterc\ztest$} (p2)
				;		
			\end{tikzpicture}
			\caption{System $\machine_6$ is branch-WSTS.} \label{fig:counter_cov1}
		\end{figure}

		In Figure~\ref{fig:counter_cov1}, let us consider the counter machine $\machine_6$ with a single counter~$c$. Let $x = (q_1,0)$. We see that ${\Pre^*(\upclose x) }= \{(q_1, n) \mid n\geq 0\} \cup \{(q_0, 0)\}$. However, ${\upclose \Pre^*(\upclose x)} = {\Pre^*(\upclose x)} \cup \{(q_0, n) \mid n\geq 1\}$. Hence, 
			given a branch-effective branch-WSTS $\system = (X, \Sigma, \srun, \leq, x_0)$ and a state $x \in X$, the set $\Pre^*(\upclose x)$ is not necessarily upward-closed. Hence, we cannot use the backward algorithm to verify the coverability problem for branch-\WSTS correctly.

		Let us consider using the forward algorithm instead, where the first procedure (\cf Procedure~\ref{proc:cov}) checks for coverability, and the second procedure (\cf Procedure~\ref{proc:noncov}) looks for a witness of non-coverability. 
		
		The second procedure computes all sets $X$ which satisfy the property $\downclose \Post^*(X) \subseteq X$. This is because, for WSTS, the set $\downclose \Post^*(x)$ satisfies this property. However, we now show a counter-example of a branch-WSTS which does not satisfy this property.
		
		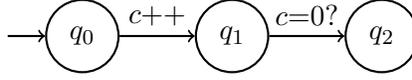
\begin{figure}[ht]
			\centering
			\begin{tikzpicture}[->, auto, thick]
				\node[initial left, initial text=, state] (a0) {$q_0$};
				\node[state] (a1) [right= of a0]{$q_1$};
				\node[state] (a2) [right= of a1]{$q_2$};
			
				\path[->]
				(a0) edge node[] {$\counterc\incr$} (a1)
				(a1) edge node[] {$\counterc\ztest$} (a2)
			
				;
			\end{tikzpicture}
			\caption{System $\machine_7$ is branch-WSTS.}
			\label{fig:branwsts2}
		\end{figure}

		Consider the counter machine  $\machine_7$ from Figure~\ref{fig:branwsts2}, with $x_0 = (q_0,0)$. We compute $Y = \downclose \Post^*(x_0)$. We see that $\Post^*(x_0) = \{ (q_0, 0), (q_1,1)\}$, hence, $Y = \downclose \Post^*(x_0) = \{ (q_0, 0), (q_1,1), (q_1, 0)\}$. However,  
		as $\downclose \Post^*(Y) =  \{ (q_0, 0), (q_1,1), (q_1, 0), (q_2, 0) \}$ is strictly larger than $Y$, we have $\downclose \Post^*(Y) \not\subseteq Y$.
			Therefore, for branch-effective, branch-WSTS $\system = (X, \Sigma, \srun, \leq, x_0)$ such that $\downclose \Post(\downclose x)$ is computable for all $x \in X$, the set $Y = \downclose \Post^*(x_0)$ does not necessarily 
			satisfy the property  $\downclose \Post^*(Y) \subseteq Y$. Hence, it is not possible to guarantee the termination of the forward coverability algorithm.

		We can deduce:
		
		\begin{prop}
			For branch-WSTS, both the backward coverability algorithm and the forward coverability algorithm do not terminate in general.
		\end{prop}
		
		Not only the two coverability algorithms do not terminate but we may prove that coverability is undecidable.
		
		\begin{thm}
			The coverability problem is undecidable for branch-effective branch-WSTS $\system = (X, \Sigma, \srun, \leq, x_0)$ (even if 
			$\system$ is strongly monotone and $\leq$ is \wqo). 
		\end{thm}
		
		\begin{proof}
			We use the family of systems given in the proof of Theorem 4.3 \cite{finkel_well-structured_2004}. Let us denote by $\mathit{TM}_j$ the $j^{th}$ Turing Machine in some enumeration. Consider the family of functions $f_j :  \mathbb{N}^2 \srun  \mathbb{N}^2$ defined by $f_j(n,k)=(n,0) $ if $k=0$ and $TM_j$ runs for more than $n$ steps, else  $f_j(n,k)=(n,n+k)$.

			Let $g :  \mathbb{N}^2 \srun  \mathbb{N}^2$ be the function defined by $g(n,k)=(n+1,k)$. The transition system $\system_j$ induced by the two functions $f_j$ and $g$ is strongly monotone hence it is also branch-monotone. Moreover, system $\system_j$ is branch-effective and we observe that $\Post$ is computable and $\leq$ is \wqo. Now, we have $(1,1)$ is coverable from $(0,0)$ in $\system_j$ \ifff $\mathit{TM}_j$ halts. This proves that coverability is undecidable.
		\end{proof}
		
		\smallskip
		\noindent\textbf{Decidability of coverability.}
		We show that coverability is decidable for a class of systems with a \wqo but with a restricted notion of monotony.
		As for example in \cite{blondin_handlin_2018}, we define the $\Cover$ of a system $\system$ from a state $x$ by: $\Cover_\system(x)=\downclose \Post_\system^*(x)$. Let us consider the following monotony condition.

		\begin{defi} 
			Let  $\system = (X, \Sigma, \srun, \leq, x_0)$ be a system. We say that $\system$ is \emph{cover-monotone} (resp. strongly cover-monotone) if, for all $y_1 \in \Cover_\system(x_0)$ and for all $x_1,x_2 \in X$ such that $x_1 \leq y_1$ and $x_1 \srun x_2 $, there exists a state $y_2 \in X$ such that $y_1 \srunp{*} y_2$ (resp. $y_1 \srun y_2$) and $x_2 \leq y_2$. 
			
		\end{defi}
		
		Let us emphasise that cover-monotony of a system $\system = (X, \Sigma, \srun, \leq, x_0)$ is a property that depends on the initial state $x_0$ while the usual monotony does not depend on any initial state (see Figure~\ref{fig:covx-notcovy}).

		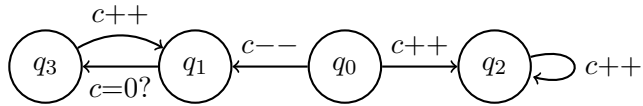
\begin{figure}[b]
			\centering
			\begin{tikzpicture}[->, auto, thick]
				\node[state] (a0) []{$q_0$};
				\node[state] (a1) [left= of a0]{$q_1$};
				\node[state] (a2) [right= of a0]{$q_2$};
				\node[state] (a3) [left= of a1]{$q_3$};
				
				\path[->]
				(a1) edge node[] {$\counterc\ztest$} (a3)
				(a3) edge[bend left] node[] {$\counterc\incr$} (a1)
				(a0) edge node[swap] {$\counterc\decr$} (a1)
				(a0) edge node[] {$\counterc\incr$} (a2)
				(a2) edge [loop right] node[swap] {$\counterc\incr$} (a2)
				;
			\end{tikzpicture}
			\caption{Machine $\machine_8$ is cover-monotone. However, if we modify the system such that the initial state $(q_0, 1)$, then it is not cover-monotone.}
			\label{fig:covx-notcovy}
		\end{figure}

	\begin{rem}
		The strong cover-monotony property is not trivially decidable for general models while (usual) strong-monotony is decidable for many powerful models like FIFO machines and counter machines. However, this notion is still of theoretical interest, as it shows that we can relax the general monotony condition.
	\end{rem}

	However, there is a link between general monotony and cover-monotony.
	
	\begin{restatable}{prop}{CoverxComp} \label{prop:CoverxComp}
		A system $\system = (X, \Sigma, \srun, \leq)$ is monotone \ifff for all $x_0\in X$,  \break $(X, \Sigma, \srun, \leq, x_0)$ is cover-monotone.
	\end{restatable}
	\begin{proof}
		Every monotone system is trivially cover-monotone for all $x_0 \in X$. 
		Conversely, consider a system $\system = (X, \Sigma, \srun, \leq)$ such that $(X, \Sigma, \srun, \leq, x_0)$ is cover-monotone for all $x_0\in X$. Let us consider $x_1,y_1,x_2 \in X$ such that $x_1 \leq y_1$ and $x_1 \srun x_2$. In order to show that $\system$ is monotone, we need to prove that there exists $y_2 $ such that $y_1 \srunp{*} y_2$ and $x_2 \leq y_2$. 

		Since $x_1 \leq y_1$ (by hypothesis), $x_1 \in \Cover(y_1)$.  By the hypothesis,
		$(X, \Sigma, \srun, \leq, y_1)$  is cover-monotone, hence
		there exists $y_2$ such that $y_1 \srunp{*} y_2$ with $x_2 \leq y_2$. Hence, $\system$ is monotone.
	\end{proof}
	
	We may now define cover-WSTS as follows.
	
	\begin{defi}
		A \emph{cover-WSTS} is a finitely branching cover-monotone system $\system = (X, \Sigma,$ $\srun, \leq, x_0)$ such that $(X,\leq)$ is \wqo.
	\end{defi}

	For cover-WSTS, the backward algorithm fails. This is once again because the presence of a single zero test removes the property of the set being upward-closed. But we will now show that the forward coverability approach is possible.

	\begin{restatable}{prop}{invariantPost} \label{prop:invariant-post}
		Given a system $\system = (X, \Sigma, \srun{}, \leq, x_0)$ and a downward-closed set $D \subseteq X$ such that $\downclose \Post(D) \subseteq D$, then we have the inclusion $\downclose \Post^*(D) \subseteq D$.
	\end{restatable}
	\begin{proof}
		We prove the following claim by induction first.
		\smallskip
		\noindent \emph{Claim:~ Given a system $\system = (X, \Sigma, \srun, \leq, x_0)$, for every downward-closed set $D \subseteq X$ such that $\downclose \Post(D) \subseteq D$, the following inclusion holds: $\downclose \Post^k(D) \subseteq D$, for all $k \geq 1$.}
		\smallskip
		\noindent Base case: For $k=1$, this is the hypothesis. 
		\smallskip
		\noindent Inductive step: The inductive hypothesis asserts that: Given a system $\system =  (X, \Sigma, \srun, \leq, x_0)$, for every downward-closed set $D \subseteq X$ such that $\downclose \Post(D) \subseteq D$, the following inclusion holds: $\downclose \Post^k(D) \subseteq D$. 
		\smallskip
		\noindent We now prove that it is also true for $k+1$.
		
		\noindent Let $ Z = \Post^k(D)$. Since $Z \subseteq \downclose \Post^k(D)$, we know that $Z \subseteq D$ (by hypothesis). Furthermore, for any subset $Y \subseteq D$, $\downclose \Post(Y)$ is also a subset of $D$. Therefore, $\downclose \Post(Z) \subseteq D$. Hence, we deduce $\downclose \Post( \Post^k(D)) \subseteq D$, i.e. $\downclose \Post^{k+1}(D) \subseteq D$. Hence, we have proved the claim by induction.
		
		From the above claim, we know that $\downclose \Post^k(D) \subseteq D$ for all $D \subseteq X$.
		Note also that $\Post^*(D) = D \cup \bigcup\limits_{k \geq 1} \Post^k(D)$. Therefore, $\Post^*(D) \subseteq D$, and finally since $D$ is  \break downward-closed,
		$\downclose \Post^*(D) \subseteq D$.
	\end{proof}

	Let us define a particular instance of the coverability problem in which we verify if a state is coverable from the initial state.

	\begin{defi}
		Given a system $\system= (X, \Sigma, \srun, \leq, x_0)$. \emph{The $x_0$-coverability problem} is: Given a state $y \in X$, do we have $y \in \downclose \Post_\system^*(x_0)$\,?
	
	\end{defi}
	
	We show that $x_0$-coverability is decidable for cover-WSTS:
	
	\begin{thm}
	
		Let  $\system = (X, \Sigma, \srun, \leq, x_0)$ be an ideally effective cover-WSTS such that $\Post$ is computable. Then, the $x_0$-coverability problem is decidable.
	
	\end{thm}

	\begin{proof}
		Consider a system $\system = (X, \Sigma, \srun, \leq, x_0)$ that is cover-WSTS, and let us consider a state $y \in X$. 
		To find a certificate of coverability (if it exists),
		we cannot use Procedure~\ref{proc:cov} since general monotony is not satisfied and then, in general, $\downclose \Post^*(x_0) \neq \downclose \Post^*(\downclose x_0)$ but we can use a variation of Procedure~\ref{proc:cov}, where we iteratively compute $x_0$, $\Post(x_0)$, $\Post(\Post(x_0))$, and so on, and at each step check if $y \leq x$ for some $x$ in the computed set. This can be done because $\system$ is finitely branching and the sets $\Post^k(x_0)$ are computable for all $k \geq 0$. Hence, if there exists a state that can cover $y$ reachable from $x_0$, it will eventually be found.

		Now, let us prove that Procedure~\ref{proc:noncov} terminates for input $y$ \ifff $y$ is not coverable from $x_0$.
	
		If Procedure~\ref{proc:noncov} terminates, then at some point, the \texttt{while} condition is not satisfied and there exists a set $D$ such that $y \notin D$ and $x_0 \in D$ and $\downclose \Post(D) \subseteq D$. Moreover, $\downclose \Post^*(I) \subseteq I$ for every inductive invariant $I$ (see Proposition~\ref{prop:invariant-post}). Hence, $\Cover_\system(x_0) \subseteq D$, therefore, since $y \notin D$, we deduce that $y \not\in \Cover_\system(x_0)$ and then $y$ is not coverable from $x_0$.
		
		Note that every downward-closed subset of $X$ decomposes into finitely many
		ideals since $(X, \leq)$ is \wqo. Moreover, since $\system$ is ideally effective, ideals of $X$ may be effectively enumerated. By \cite{blondin_handlin_2018} and \cite{blondin_well_2017}, for ideally effective systems, testing of inclusion of downward-closed sets, and checking the membership of a state in a downward-closed set, are both decidable.

		To show the opposite direction, let us prove that if $y$ is not coverable from $x_0$, the procedure terminates. It suffices to prove that $\Cover_\system(x_0)$ is an inductive invariant. Indeed, this implies that $\Cover_\system(x_0)$ is eventually computed by Procedure~\ref{proc:noncov} when $y$ is not coverable from $x_0$. 
		
		Let us show $\downclose \Post(\Cover_\system(x_0)) \subseteq \Cover_\system(x_0)$. Let $b \in \downclose \Post(\Cover_\system(x_0))$. Then, there exists $a, a', b'$ such that $x_0 \srunp{*} a'$, $a' \geq a$, $a \srun b'$ and $b' \geq b$. Furthermore, $a', a \in \Cover(x_0)$. Hence, by cover-monotony, there exists $b'' \geq b'$ such that $a' \srunp{*} b''$. 
		Therefore, $x_0 \srunp{*} b''$ and $b'' \geq b' \geq b$, hence, $b \in \Cover_\system(x_0)$. 
		Hence, the $x_0$-coverability problem is decidable.
	\end{proof}
	On the other hand, we cannot decide the (general) coverability problem for this class of systems:

	\begin{thm} 
		The coverability problem is undecidable for cover-WSTS.
	\end{thm}
	
	\begin{proof}
		Given any counter machine $\countermachine =(Q,\counterset , T, q_0)$, let $\system_\countermachine = (X,A_\countermachine,\srun,\leq,  x_0)$ be its transition system equipped with the natural order on counters. We can construct a system $\system' = (X', A_\countermachine, \srun', \leq, x'_0)$ such that $\system'$ is  cover-monotone, and any state $x \in X$ is coverable \ifff it is also coverable in $X'$. The construction is as follows. We add a new control-state $q$ from the initial state in the counter machine ($q_0$) reachable via an empty transition, therefore, $X' = X \cup \{(q,0) \}$. This new control-state is a sink state, i.e. there are no transitions from $q$ to any other control-state (except itself). Moreover, we let $x'_0 = (q, 0)$. Note that $\system'$ is cover-monotone, because there is no state reachable from $x'_0$, hence, the property is vacuously satisfied. However, for all other states, as we leave the system unchanged, we see that a state $x$ is coverable in $\system$ by a state $y$ \ifff it is coverable in $\system'$. Hence, coverability for counter machines reduces to the coverability problem for cover-WSTS, and coverability is therefore undecidable for cover-WSTS.
	\end{proof}
	
	\section{Conclusion}\label{sec:concl}
	
	We have tried to relax the notions of monotony and of the wellness of the quasi-ordering which were traditionally used to define a WSTS. We observed that we do not need the wellness of the quasi-ordering or monotony between all states. By relaxing the conditions to only states reachable from one another, thus defining what we call \emph{branch-WSTS}, we are still able to decide non-termination and boundedness. Moreover, some systems that have been studied recently have been shown to belong to this class, which adds interest to this relaxation. Furthermore, we have used the ideas of branch-monotony to hold along a single run, in order to provide an algorithm that can decide non-termination for a larger class of counter machines and branch-\wqo FIFO machines.

	However, as coverability is undecidable for branch-WSTS, the notion of coverability seems to require a stricter condition than what we define for branch-WSTS. This leads us to introduce a different class of systems, incomparable to branch-WSTS, which we call \emph{cover-WSTS}. These systems relax the condition of monotony to only states within the coverability set, while still retaining the decidability of a restricted form of coverability.
	
	As future work, other systems that belong to these classes can be studied. Moreover, it would be interesting to precisely characterise the systems for which we can decide non-termination and boundedness along a single run. It would also be interesting to see if the branch-WSTS relaxation translates to better hope for usability of WSTS as a verification technique.
	
	\subsubsection*{Acknowledgements} We thank the reviewers for the detailed comments and suggestions. The work reported was carried out in the framework of ReLaX, UMI2000 (ENS Paris-Saclay, CNRS, Univ. Bordeaux, CMI, IMSc). It is partly supported by ANR FREDDA (ANR-17-CE40-0013) and ANR BRAVAS (ANR-17-CE40-0028). It is also partly supported by EPSRC STARDUST (EP/T014709/2).

	\bibliographystyle{alphaurl}
	\bibliography{bibliography}
\end{document}